\documentclass[a4paper,UKenglish,cleveref, autoref, thm-restate]{lipics-v2021}
\usepackage{mathtools}
\pdfoutput=1

\usepackage[T1]{fontenc}
\usepackage{algorithm}
\usepackage[noend]{algorithmic}
\usepackage{amsmath,amsfonts,amssymb,amsthm,pifont}
\usepackage{comment}
\usepackage{enumerate}
\usepackage{etoolbox}
\usepackage{color}
\usepackage{tikz}
\usepackage{thmtools,thm-restate}

\newcommand{\EF}[1]{\ifstrempty{#1}{\textrm{\textup{EF}}}{\textrm{\textup{EF{$#1$}}}}}

\newcommand{\efx}{\textrm{EFX}}

\renewcommand{\algorithmiccomment}[1]{\bgroup\hfill\small{$\blacktriangleright$ #1}\egroup}

\newcommand{\unalloc}{\textsf{U}}
\newcommand{\fav}{\textsf{Fav}}
\newcommand{\flex}{\textsf{Flex}}
\newcommand{\res}{\textsf{Res}}
\newcommand{\zero}{\textsf{Zero}}
\newcommand{\bad}{\textsf{Bad}}

\newtheorem{clm}[theorem]{Claim}

\nolinenumbers

\title{EF1 Allocations for Identical Trilean and Separable Single-Peaked Valuations}

\titlerunning{EF1 Allocations for Identical Trilean and Separable Single-Peaked Valuations}

\author{Umang Bhaskar}{Tata Institute of Fundamental Research, India}{umang@tifr.res.in}{}{}

\author{Gunjan Kumar}{Indian Institute of Technology Kanpur, India}{gunjan@cse.iitk.ac.in}{}{}

\author{Yeshwant Pandit}{Tata Institute of Fundamental Research, India}{yeshwant.pandit@tifr.res.in}{}{}

\author{Rakshitha}{Indian Institute of Technology Delhi, India}{rakshitha.mt121@maths.iitd.ac.in}{}{}

\authorrunning{Bhaskar, Kumar, Pandit, and Rakshitha}



\acknowledgements{The first and third author acknowledge support from the Department of Atomic Energy, Government of India, under project no. RTI4001. Part of the work was done when the second author was at NUS, and fourth author was visiting TIFR.}

\hideLIPIcs

\begin{document}

\maketitle

\begin{abstract}
In the fair division of items among interested agents, envy-freeness is possibly the most favoured and widely studied formalisation of fairness. For indivisible items, envy-free allocations may not exist in trivial cases, and hence research and practice focus on relaxations, particularly envy-freeness up to one item (EF1), which allows the removal of an item to resolve envy. A significant reason for the popularity of EF1 allocations is its simple fact of existence, which supports research into combining EF1 with other properties such as efficiency, or research into stronger properties such as ex-ante envy-freeness. It is known that EF1 allocations exist for two agents with arbitrary valuations; agents with doubly-monotone valuations; agents with Boolean valuations; and identical agents with negative Boolean valuations. In fact the last two results were shown for the more stringent EFX fairness requirement. 

We consider two new but natural classes of valuations, and partly extend results on the existence of EF1 allocations to these valuations. Firstly, we consider trilean valuations --- an extension of Boolean valuations --- when the value of any subset is 0, $a$, or $b$ for any integers $a$ and $b$. Secondly, we define separable single-peaked valuations, when the set of items is partitioned into types. For each type, an agent’s value is a single-peaked function of the number of items of the type. The value for a set of items is the sum of values for the different types. We prove EF1 existence for identical trilean valuations for any number of agents, and for separable single-peaked valuations for three agents. For both classes of valuations, we also show that EFX allocations do not exist.

\end{abstract}

\section{Introduction}
\label{sec:Introduction}

Fair division refers to the fundamental problem of allocating a set of items --- called manna --- \emph{fairly} among a set of agents. The items to be allocated could be as diverse as seats in courses in a university, items in a contested inheritance or divorce proceedings, carbon credits among nations, chores among members of a household, cab fare for a shared ride, household rent among roommates, etc. These problems are frequent and universal and have naturally attracted a lot of attention from researchers in various fields.

Given the widespread applications, a number of different formalisations of what it means to be fair have naturally been proposed. A popular and possibly predominant formalisation is \emph{envy-freeness}, which informally requires that each agent prefers her own allocated manna over the allocation to any other agent. When items are indivisible and have to be wholly allocated to a single agent, envy-free allocations may not exist, and hence relaxations are studied. The most prevalent relaxation is \emph{envy-freeness upto one item (EF1)}, which allows envy among agents as long as this envy is eliminated by removing a single item~\cite{B11combinatorial,LMM+04approximately}. Envy-free allocations are a focus of theoretical research and also implemented in practical tools, e.g., spliddit.org~\cite{GP15spliddit,Shah17}, and fairoutcomes.com~\cite{FairOutcomes}. 

In fair division, given a set $M$ of items to be allocated among $n$ agents, we assume that each agent $i$ has a valuation function $v_i: 2^M \rightarrow \mathbb{Z}$ that specifies a value for each subset of items. An allocation $A = (A_1, \ldots, A_n)$ is a partition of $M$ where agent $i$ gets the set $A_i$. An allocation is EF1 if, whenever $v_i(A_i) < v_i(A_j)$, there is an item $x \in A_i \cup A_j$ so that $v_i(A_i \setminus \{x\}) \ge v_i(A_j \setminus \{x\})$, i.e., on removing item $x$, agent $i$ weakly prefers her own allocation to $j$'s allocation. 

For a formalisation such as EF1 to be practically relevant, a basic criterion is existence. If a notion of fairness is not easily satisfied, or does not exist in many instances, it's practical use is limited. EF1 satisfies this criteria in very broad classes of valuations. Further, the robustness of EF1, as well as algorithms for obtaining such allocations, have encouraged research into broader objectives. For example, researchers have studied allocations that are EF1 satisfying other properties such as Pareto efficiency~\cite{BKV18Finding,ACI+19fair,GargMQ22}, strategy-proofness~\cite{amanatidis2017,amanatidis2023allocating}), 
and ex-ante envy-freeness~\cite{aziz2020simultaneously,FSV20best}). 

Initial results established the existence of EF1 allocations for monotone valuations of the agents~\cite{LMM+04approximately}. A different algorithm is known for nonmonotone but additive valuations when each agent has a positive or negative value for each item and values a subset at the sum of the individual item values~\cite{ACI+21fair,ACI+19fair}. These were extended to the large class of doubly monotone valuations, where each agent partitions the items into goods and chores; the goods always have nonnegative marginal value for the agent, while chores always have nonpositive marginal values~\cite{BSV21approximate}. While this is a broad class, one particular property not captured by doubly monotone valuations is when ``too much of a good thing is bad.'' Doubly monotone valuations require an agent's good to always be a good, no matter what other items the agent is allocated. But this is often not the case. An hour or two of housework may be enjoyable, but more may turn this into a chore. Conversely, the initial hours spent learning a new skill are tedious, but sustained hours of practice may lead to expertise and enjoyment. A balanced diet requires precise ranges of nutrients, and both more and less may be harmful. Moulin also provides an example of such valuations~\cite{moulin2019fair}.



One would ideally like to establish EF1 existence for arbitrary valuations. A possible approach to this goal is to consider discretized valuations, i.e., to restrict the range of possible values for any bundle to a smaller set, and obtain results on the existence of EF1. In this direction, prior work shows that EF1 allocations exist if (a) all agents have values that are in $\{0,1\}$ for all sets, or (b) agents are identical, and have values in $\{0,-1\}$ for all sets~\cite{BBB+24}.\footnote{More specifically, existence is shown for the more restricted class of EFX allocations.} 

In this paper, we extend these results on the existence of EF1 allocations in two directions. 

\emph{Firstly}, we show that EF1 always exists for identical \emph{trilean} valuations --- when agents are identical, and the value for each set of items is $0$, $a$, or $b$ for any integers $a$, $b$. Our work builds on prior work on Boolean $\{0,1\}$ and $\{0,-1\}$ valuations. This suggests that considering discretized valuations may be a fruitful approach to prove the existence of EF1 for arbitrary (though identical) valuations. Along the way, we extend prior work to show that EF1 allocations also exist for non-identical $\{0,-1\}$-valuations.

While theoretically, trilean valuations are a stepping stone to arbitrary nonmonotone valuations, practically these are also an intuitive way to convey approval or disapproval of a set of items. Approval can be interpreted as value $+1$ for the set, disapproval as value $-1$, and no feedback as value $0$. 

The restriction to identical valuations is limiting. However identical valuations are often a basic tool in developing algorithms for nonidentical valuations. E.g., the existence of EFX allocations for 2 agents follows immediately from the existence for identical valuations~\cite{PR20almost}. The simplified algorithm for EFX for 3 agents starts off with an EFX allocation for identical agents~\cite{Akrami2023EFX}. There are numerous such examples.

\emph{Secondly}, we introduce a new class of valuations that we call \emph{separable single-peaked} (SSP) valuations. For SSP valuations, the set of items are partitioned into $t$ types. For each type of item $j$, each agent has a threshold $\theta_{ij}$. Agent $i$'s valuation for type $j$ is single-peaked with peak $\theta_{ij}$: it monotonically increases with the number of items up to $\theta_{ij}$, and monotonically decreases after that.  The valuation is additive across items of different types. Thus, these are a relaxation of separable piece-wise linear concave valuations, widely studied in fair division and market equilibria (e.g.,~\cite{ChaudhuryCGGHM22,GargMSV15}). For SSP valuations, we show two results: we show that EF1 allocations exist either when agents have the same threshold for a type, and when agents have different thresholds, but there are three agents. Finally, we give a tight example to show that EFX allocations do not exist for the valuations studied, even for two identical agents and three items.

One of our goals in studying these valuation classes is to go beyond doubly monotone valuations. This is uncharted territory --- we know of no other properties (other than EF1) that hold in such general valuation classes, and hence there have been few attempts to define structured classes beyond doubly monotone.\footnote{A notable exception is the class of nonmonotone submodular valuations, of which SSP is a subclass.} We think the study of structured nonmonotone valuations should be of interest in its own right.

\section{Related Work}

Fair division has traditionally focused on allocating divisible resources, also known as cake-cutting. 
 A survey on computational results for cake-cutting is presented by Procaccia~\cite{P15cake}. For non-monotone valuations (sometimes called ``burnt cake''), results are known only when the number of agents is either 4 or a prime number~\cite{S18fairly,MZ19envy}.


For indivisible manna, for monotone non-decreasing valuations --- when all items are goods --- the existence of EF1 allocations was given by Lipton et al.~\cite{LMM+04approximately}. For additive non-monotone utilities a double round-robin algorithm for EF1 allocations is known~\cite{ACI+19fair,ACI+21fair}. These results were extended to doubly monotone valuations~\cite{BSV21approximate} by suitably modifying the envy-cycle elimination algorithm of Lipton et al. 

EF1 allocations always exist for two agents, with completely general valuations~\cite{BBB+24}. Further, EFX allocations --- a stricter fairness condition, that requires that if $i$ envies $j$, this envy be resolved by removing any good from agent $j$, and any chore from agent $i$ --- exist (a) for Boolean valuations, i.e., $v_i(S)  \in \{0,1\}$ for any agent $i$ and set $S$ of items, and (b) for identical and negative Boolean valuations, i.e., $v_i(S) = v(S) \in \{0,-1\}$ for all $i$, $S$. An EFX allocation is also an EF1 allocation. We note that these existence results are through a variety of different techniques --- envy-cycle elimination, round-robin, sequential allocation of minimal subsets, local search, etc. It remains open if an EF1 allocation exists for arbitrary valuations, even, e.g., for the case of 3 identical agents.

EF1 is also considered alongside Pareto-optimality (PO), where an allocation is PO if no allocation gives at least as much utility to every agent and strictly higher utility to at least one agent. It is known that an EF1 and PO allocation always exists for additive goods and can be computed in pseudopolynomial time~\cite{CKM+19unreasonable,BKV18Finding,MurhekarG21}. For additive chores and mixed items, partial results are known~\cite{ACI+19fair,GargMQ22,hosseini2024almost}. Researchers have also studied the existence of the stricter EFX allocations with monotone additive valuations. These are known to exist if all agents are identical~\cite{PR18almost,PR20almost}, have matroid rank valuations~\cite{BabaioffEF21}, or if there are three agents~\cite{Akrami2023EFX,ChaudhuryGM20}. 
The case of four or more agents has partial results~\cite{BergerCFF22,GhosalPNV23}.



Besides envy-freeness, numerous other fairness notions are also studied in the literature, including proportionality, equitability, and maximin share. Amanatidis et al.~\cite{amanatidis2023survey} and Mishra, Padala, and Gujar~\cite{MishraPG23} present surveys on recent developments on these.

\section{Preliminaries}
\label{sec:Preliminaries}

\paragraph*{Problem instance}
An \emph{instance} $\langle N, M, V \rangle$ of the fair division problem is defined by a set $N$ of $n \in \mathbb{N}$ \emph{agents}, a set $M$ of $m \in \mathbb{N}$ \emph{indivisible items}, and a \emph{valuation profile} $V = \{v_1,v_2,\dots,v_n\}$ that specifies the preferences of every agent $i \in N$ over each subset of the items in $M$ via a \emph{valuation function} $v_i: 2^{M} \rightarrow \mathbb{R}$. Agents are identical if $v_i(S) = v(S)$ for all agents $i \in N$ and all subsets $S \subseteq M$. 

\paragraph*{Allocation}
An \emph{allocation} $A \coloneqq (A_1,\dots,A_n)$ is an $n$-partition of a subset of the set of items $M$, where $A_i \subseteq M$ is the \emph{bundle} allocated to the agent $i$ (note that $A_i$ can be empty). An allocation is said to be \emph{complete} if it assigns all items in $M$, and is called \emph{partial} otherwise.

\paragraph*{Envy-Freeness and its Relaxations}
An allocation is \emph{envy-free} if for all agents $i$, $j$, $v_i(A_i) \ge v_i(A_j)$, i.e., each agent prefers their own bundle to that of any other agent's. An allocation is \emph{envy-free upto one item} (EF1) if, whenever $v_i(A_i) < v_i(A_j)$, there is some item $x \in A_i \cup A_j$ so that $v_i(A_i \setminus \{x\}) \ge v_i(A_j \setminus \{x\})$. Given an allocation $A$, we say a subset $N'$ of agents is \emph{mutually EF1} if the EF1 condition holds for every pair of agents in $N'$. We introduce additional notation specific to the valuations we study in the relevant sections.

\paragraph*{Trilean and Boolean Valuations}
For a fair division instance, valuations are trilean if for some $a$, $b \in \mathbb{Z}$, for every agent $i \in N$ and $S \subseteq M$, $v_i(S) \in \{0,a,b\}$. Valuations are Boolean $\{0,1\}$-valued if for every agent $i$ and subset $S$ of items, $v_i(S) \in \{0,1\}$. Similarly, valuations are Boolean $\{0,-1\}$-valued if for every agent $i$ and subset $S$ of items, $v_i(S) \in \{0,-1\}$. 

\paragraph*{Separable Single Peaked Valuations}
We first define $\theta_{ij}$ for $i \in N$ and $j \in [t]$ as the threshold of agent $i$ for items of type $j$. Then agent $i$'s valuation $v_i(A_i) = \sum_{j=1}^t v_{ij}(a_{ij})$, where the valuations $v_{ij}$ are single-peaked: for all $x \le y \le \theta_{ij}$, $v_{ij}(x) \le v_{ij}(y)$, while for $\theta_{ij} \le x \le y$, $v_{ij}(x) \ge v_{ij}(y)$.

\paragraph*{Envy graph}
The \emph{envy graph} $G_A$ of an allocation $A$ is a directed graph on the vertex set $N$ with a directed edge from agent $i$ to agent $k$ if $v_i(A_k) > v_i(A_i)$, i.e., if agent $i$ prefers the bundle $A_k$ over the bundle $A_i$.

\paragraph*{Top-trading envy graph}
The \emph{top-trading envy graph} $T_A$ of an allocation $A$ is a subgraph of its envy graph $G_A$ with a directed edge from agent $i$ to agent $k$ if $v_i(A_k) = \max_{j \in N} v_i(A_j)$ and $v_i(A_k) > v_i(A_i)$, i.e., if agent $i$ envies agent $k$ and $A_k$ is the most preferred bundle for agent $i$.

\paragraph*{Cycle-swapped allocation}
Given an allocation $A$ and a directed cycle $C$ in an envy graph or a top-trading envy graph, the \emph{cycle-swapped allocation} $A^C$ is obtained by reallocating bundles backwards along the cycle. For each agent $i$ in the cycle, define $i^+$ as the agent that she is pointing to in $C$. Then, $A_i^C = A_{i^+}$ if $i \in C$, otherwise $A_i^C = A_i$.

\section{EF1 Allocations for Trilean Valuations}
\label{sec:EF_Trilean_Valued}

To show EF1 exists for identical trilean valuations, we claim that it is sufficient to prove existence for two cases: $a = 1$, $b = 2$, and $a=-1$, $b=1$. Since for envy-freeness we only compare relative values of bundles, this is immediate if either $a$ or $b$ is nonnegative (and in this case holds for non-identical trilean valuations). In the appendix, we prove this also holds if $a$ and $b$ are both negative.

\begin{restatable}{proposition}{rstab}
Suppose an EF1 allocation exists in all instances $(N,M,\mathcal{V})$ with identical agents, where either $v(S) \in \{0,1,-1\}$ for all $S \subseteq M$, or $v(S) \in \{0,1,2\}$. Then an EF1 allocation exists in all instances with identical agents where $v(S) \in \{0,a,b\}$ for any integers $a$, $b$.
\label{prop:ab}
\end{restatable}

Our proof for trilean valuations will thus focus on these two cases, called negative trilean if values are in $\{0,-1,1\}$, and positive trilean if values are in $\{0,1,2\}$.


Since EF1 studies values of sets upon removal of items, we introduce some notation for this. For a set of items $S$, any immediate subset ($S' \subset S$ s.t. $|S'| = |S| - 1$) is called a \emph{child} of $S$.

\begin{definition}
For an agent $i$ and a subset $S$ of items, and $a,b \in \mathbb{Z}$,
\begin{itemize}
\item We use $v_i(S) = a \rightarrow b$ to denote that $v_i(S) = a$, and for \emph{some} item $x \in S$, $v_i(S \setminus \{x\}) = b$. 
\item  We use $v_i(S) = a \rightrightarrows b$ to denote that $v_i(S) = a$, and for \emph{every} $x \in S$, $v_i(S\setminus\{x\}) = b$. 
\item For $B \subset \mathbb{Z}$, we use $v_i(S) = a \rightrightarrows B$ to denote that $v_i(S) = a$, and for every $x \in S$, $v_i(S \setminus \{x\}) \in B$. 
\end{itemize}
\end{definition}

\noindent For example, $v_i(S) = 0 \rightarrow 1$ denotes that $v_i(S) = 0$, and $\exists x \in S$ for which $v_i(S \setminus \{x\}) = 1$. The notation $v_i(S) = -1 \rightrightarrows \{-1,0\}$ denotes that $v_i(S) = -1$, and on removal of any $x \in S$, $v_i(S \setminus \{x\})$ is either $-1$ or $0$.

\paragraph*{A note on trilean valuations} 

In our work, we use \emph{trilean} for instances where any set of items has one of three values.\footnote{The term trilean was suggested by attendees of Dagstuhl seminar 24401.} The term is a natural extension to the term Boolean (see, e.g.,~\cite{wiki:trilean}), which is used in prior work for instances where any set of items has one of two values~\cite{BBB+24}. It is important to distinguish the valuations we study from the case where any \emph{item} has $0$ or $1$ \emph{marginal} value (and thus a set of items can take any value between $0$ and $m$). Various other terms have also been used for this case, including dichotomous~\cite{BabaioffEF21}, binary~\cite{HalpernPPS20,BKV18greedy} and bivalued instances~\cite{GargMQ22}. Given that there is a lack of consistent notation, and that we are the first to propose and study such valuations, we believe the use of trilean as an extension of Boolean to denote the valuations we study is a reasonable choice.

\subsection{Boolean Valuations}
\label{sec:binaryVals}

As noted, B\'erczi et al. show that for Boolean $\{0,1\}$ valuations, there exists an EFX allocation, and give an algorithm for this~\cite{BBB+24}. They also give an algorithm for obtaining an EFX allocation for \emph{identical} Boolean $\{0,-1\}$ valuations. Since EFX is a stronger requirement than EF1, these algorithms give EF1 allocations for the respective cases. This however leaves open the existence of EF1 allocations for nonidentical Boolean $\{0,-1\}$ valuations.\footnote{For identical valuations, there is a reduction from finding an EF1 allocation in Boolean $\{0,-1\}$ valuations to finding one in Boolean $\{0,1\}$ valuations --- replace each $-1$ value with $+1$, and find an EF1 allocation $A$ in the resulting $\{0,1\}$-valued instance. Then $A$ is an EF1 allocation in the original $\{0,-1\}$-valued instance as well. This reduction, however, does not work for \emph{nonidentical} valuations.} We now give such an algorithm, called NegBooleanEF1. Our algorithm is similar to the algorithm for Boolean $\{0,1\}$ valuations, with suitable modifications for negative valuations. We use certain properties of the allocation produced by this algorithm later, in our result for trilean valuations (Proposition~\ref{prop:binZeroMOne2}).



For a brief intuition for our algorithm, it can be checked that the only scenario when an allocation is not EF1 is if we have two agents $i$ and $j$ such that $v_i(A_i)=-1 \rightrightarrows -1$, and $A_j=\emptyset$ or $v_i(A_j)=0\rightrightarrows 0$. Then, agent $i$ envies agent $j$, but the envy is not eliminated by the removal of any item.

To avoid this, from the set $M'$ of unallocated items, Algorithm NegBooleanEF1 repeatedly chooses an inclusionwise minimal set $S$ for which all unallocated agents have value $-1$. This set $S$ is assigned to an agent $i$ so that $v_i(S) = -1 \rightarrow 0$. Since $S$ is minimal, such an agent must exist. The main property is that this agent cannot be part of an EF1 violation: no later agent envies agent $i$, since all unallocated agents have value $-1$ for set $S$. Also any envy agent $i$ has can be eliminated by removing an item $x$ such that $v_i(S \setminus \{x\}) = 0$. If there are no more such minimal sets, then for some agent $i$, $v_i(M') = 0$. We assign $M'$ to agent $i$ and complete the allocation.

\begin{algorithm}[ht]
\caption{NegBooleanEF1}
\label{alg:binZeroMOne}
\begin{algorithmic}[1]
\REQUIRE Fair division instance $(N,M,\mathcal{V})$  with Boolean $\{0,-1\}$ valuations.
\ENSURE An \EF{1} allocation $A$.
\STATE Initialise $A=(\emptyset,\hdots,\emptyset)$, $M' = M$, $N' = N$.
\WHILE{($M' \neq \emptyset$) AND ($|N'| \ge 2$) AND ($\forall~ i \in N',~v_i(M') = -1$) }
    \STATE Let $S \subseteq M'$ be an inclusion-wise minimal set so that  $v_i(S) = -1 ~ \forall i \in N'$ \label{cond2}
\COMMENT{For all $x \in S$, $v_i(S\setminus \{x\}) = 0$ for some $i \in N'$.}
    \STATE Pick $i \in N'$ so that $v_i(S) = -1 \rightarrow 0$. 
    \COMMENT{Since $S$ is minimal, agent $i$ exists.}
    \STATE Assign $A_i = S$, $N'=N'\setminus \{i\}$, $M' = M' \setminus S$.
\ENDWHILE
\IF[There exists $i \in N'$ s.t. $v_i(M')=0$.]{($M' \neq \emptyset$) AND ($|N'| \geq 2$)}
    \STATE Choose $i \in N'$ so that $v_i(M') = 0$. Let $A_i = M'$. \label{line:negativebooleanif}
\ELSE[Either $M' = \emptyset$ or $|N'| = 1$.]
    \STATE Choose $i \in N'$, let $A_i = M'$. \label{line:negativebooleanelse}
\ENDIF
\RETURN allocation $A$.
\end{algorithmic}
\end{algorithm}

\begin{restatable}{theorem}{rstnegativebinary}
Given a fair division instance with negative Boolean valuations, Algorithm NegBooleanEF1 returns an EF1 allocation in polynomial time.
\label{thm:negativebinary}
\end{restatable}

\begin{proof}
We reindex the agents so they are allocated in increasing order. For any agent $i$ allocated in the while loop, $v_j(A_i) = -1$ for all $j \ge i$, hence no agent after $i$ envies agent $i$. Further, $i$ is chosen so that for some $x \in A_i$, $v_i(A_i \setminus \{x\}) = 0$. Hence if agent $i$ envies another agent, this can be resolved by removing $x$. It follows that there is no EF1 violation involving an agent allocated in the while loop.

For the remainder of the algorithm, note that if Line~\ref{line:negativebooleanif} executes, then all the remaining agents have value $0$ for their bundles. Hence these agents are mutually EF1, and the allocation must then be EF1. If Line~\ref{line:negativebooleanelse} executes, then either $M' = \emptyset$, in which case the allocation is unchanged, or all of the agents $1, \hdots, n-1$ have been allocated in the while loop, and agent $n$ is allocated an arbitrary set. But the allocation remains EF1, since as stated, no EF1 violation can involve agents $1, \hdots, n-1$. It is not hard to verify that the algorithm makes a polynomial number of queries.
\end{proof}

\subsection{Negative Trilean Valuations}
\label{sec:trileanIdent}
In this section, we establish the existence of \EF{1} allocations when agents are identical and their valuations are negative trilean. We prove the following theorem in the remainder of this section.

\begin{theorem}
Every instance with identical negative trilean valuations has an EF1 allocation.
\label{theorem:trilean}
\end{theorem}

Before we give a brief description of how we achieve this, let us define a few terms. At any point in our algorithm, each agent will belong to one or more sets, depending on the bundle allocated to the agent. This classification forms the basis of our algorithm and the analysis, as it clarifies when EF1 violations occur. The conditions on EF1 violations are shown in Lemma~\ref{lem:ef1violations}.

\begin{enumerate}
\item Unallocated: $\unalloc = \{i : A_i = \emptyset\}$.
\item Zero: $\zero = \{i : v(A_i) = 0\}$.
\item Favourable: $\fav = \{i : v(A_i) = 1 \rightarrow -1 \text{ or } v(A_i) = -1 \rightarrow 1\}$.
\item Flexible: $\flex^+  = \{i: v(A_i) = 0 \rightarrow 1\}$, and $\flex^- =\{i:  v(A_i) = 0 \rightarrow -1\}$.
\item Resolved: $\res^+ =\{i : v(A_i) = 1 \rightarrow 0\}$, and $\res^- = \{i : v(A_i) = -1 \rightarrow 0\}$.
\item Bad: $\bad^+ = \{i: v(A_i) = 1 \rightrightarrows 1\}$, and $\bad^- = \{i: v(A_i) = -1 \rightrightarrows -1\}$
\end{enumerate}

Note that the above sets are not mutually exclusive (e.g., an agent could be in both $\flex^-$ and $\flex^+$, and an agent in $\unalloc$ is also in $\zero$), but are exhaustive, i.e., for any allocation $A$ and identical trilean valuations for the agents, each agent $i$ falls in one or more of the above sets. We will not explicitly move agents in and out of these sets. Rather, agents will acquire or lose membership depending on the bundle allocated to them. 

We use these terms to describe the respective sets as well. Thus a set of items $S$ is:

\begin{enumerate}
\item Zero-valued if $v(S) = 0$.
\item Favourable if $v(S) = 1 \rightarrow -1$ or $v(S) = -1 \rightarrow 1$.
\item Flexible if $v(S) = 0 \rightarrow 1$ or $v(S) = 0 \rightarrow -1$.
\item Resolved if $v(S) = 1 \rightarrow 0$ or $v(S) = -1 \rightarrow 0$.
\item Bad if $v(S) = 1 \rightrightarrows 1$ or $v(S) = -1 \rightrightarrows -1$.
\end{enumerate}

Apart from the set \fav, agents in the same set have the same value for their bundles, and hence do not envy each other. Given an allocation $A$ and a pair $i$, $j$ of agents, we now use these sets to give necessary conditions for the violation of EF1.

\begin{figure}[!ht]
\vspace{0.1in}
\centering
\begin{tikzpicture}

    \node[draw, rectangle] (X) at (2, -2) {$\fav$};
    \node[draw, rectangle] (L1) at (0, -3) {$\flex^- (0 \rightarrow -1)$};
    \node[draw, rectangle] (L2) at (0, -4) {$\res^- (-1 \rightarrow 0)$};
    \node[draw, rectangle] (L3) at (0, -5) {$\bad^- (-1 \rightrightarrows -1)$};
    \node[draw, rectangle] (R1) at (4, -3) {$\flex^+ (0 \rightarrow +1)$};
    \node[draw, rectangle] (R2) at (4, -4) {$\res^+ (1 \rightarrow 0)$};
    \node[draw, rectangle] (R3) at (4, -5) {$\bad^+ (1 \rightrightarrows 1)$};
    \node[draw, rectangle] (Y) at (2, -6) {$\zero$ $(0)$};

    \foreach \i in {1,2,3}
            \draw (L3) -- (R\i);
	\draw (L3) -- (Y);
    \foreach \i in {1,2,3}
            \draw (R3) -- (L\i);
	\draw (R3) -- (Y);
	\draw (L2) -- (R2);
\end{tikzpicture}
\caption{The edges depict possible EF1 violations between different sets of agents.}
\label{fig:ef1violations}
\end{figure}
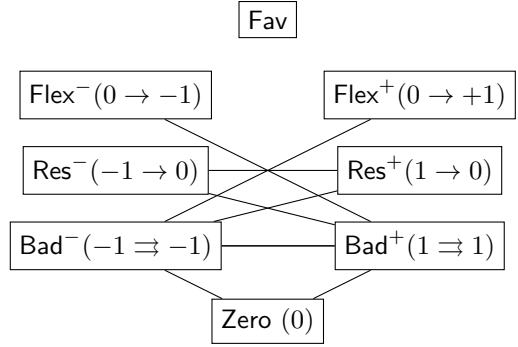

\begin{restatable}{lemma}{rstefviolations}
Given agents with identical negative trilean valuations and an allocation $A$, agents $i$, $j$ are NOT mutually EF1 only if:\footnote{Note that $\unalloc \subseteq \zero$.}
\begin{enumerate}
\item $i \in \bad^+$, and $j \in \zero \cup \flex^- \cup \res^- \cup \bad^-$, or
\item $i \in \bad^-$, and $j \in \zero \cup \flex^+ \cup \res^+ \cup \bad^+$, or
\item $i \in \res^+$ and $j \in \res^-$.
\end{enumerate}
\label{lem:ef1violations}
\end{restatable}

In particular, note that agents in $\fav$ are mutually EF1 with every other agent, regardless of the other agent's bundle. The lemma is depicted in Figure~\ref{fig:ef1violations}, with edges between various sets showing possible EF1 violations.

At a very high level, our algorithm proceeds by first allocating \emph{favourable} sets until there are no more, then by allocating \emph{flexible} sets until the remaining set of items is either zero-valued or Boolean-valued ($\{0,1\}$ or $\{0,-1\}$). If the remaining set is zero-valued, we assign it and return the resulting allocation. By Lemma~\ref{lem:ef1violations} this is EF1. If the remaining set is Boolean-valued, then we use the previous algorithms for Boolean values to allocate \emph{resolved} and \emph{zero} sets. Eventually, we will be left with possibly a single \emph{bad} set. The resolved sets and the bad set assigned will be of the same sign (either $\res^+$ and $\bad^+$, or $\res^-$ and $\bad^-$). It follows from Lemma~\ref{lem:ef1violations} that the only possible EF1 violations will be between the flexible sets and the single bad set, which we will then fix in Algorithm FixEF1ViolationsNeg.

We now describe our algorithm in more detail. We first allocate any favourable sets ($v(S) = 1 \rightarrow -1$ or $-1 \rightarrow 1$) to the first $n-1$ agents, as long as there are any favourable sets. If there are items remaining but no more favourable sets, and not all agents in $[n-1]$ are allocated, let $M'$ be the set of items remaining. If $v(M') = 0$ (or $M' = \emptyset$), we assign $M'$ to the next agent and return the resulting allocation. Else, we next try and find \emph{maximal} flexible sets. That is, if $v(M') = 1$, we look for a maximal set $S$ such that $v(S) = 0 \rightarrow -1$, and if $v(M') = -1$, we look for a maximal set $S$ such that $v(S) = 0 \rightarrow 1$. If $M'$ is trilean, such a set $S$ must exist: since there are no favourable sets, if $M'$ and $S \subset M'$ have opposite signs, there must exist a set $T$ so that $S \subset T \subset M'$ and $v(T) = 0$.

If after allocating any flexible sets to agents in $[n-1]$, at least two agents remain and $v(M') = 0$, assign $M'$ to the next agent and return the resulting allocation, which is EF1. If $v(M') \neq 0$ then $M'$ is either Boolean $\{0,1\}$-valued or Boolean $\{0,-1\}$-valued. We then call the respective algorithms for obtaining EF1 allocations for these respective cases. This completes the allocation, though there may be EF1 violations. We then call Algorithm FixEF1ViolationsNeg to fix any violations.

The remaining case is that all agents in $[n-1]$ have been allocated either favourable or flexible sets. In this case, we allocate the remaining items to agent $n$, and call Algorithm FixEF1ViolationsNeg to fix any EF1 violations.

We will use the algorithms for obtaining EF1 allocations for Boolean $\{0,1\}$ and $\{0,-1\}$ valuations as subroutines. The next two propositions state some properties of the allocations returned by these algorithms. For Boolean $\{0,1\}$-valuations, we use Algorithm 2 from B\`erczi et al.~\cite{BBB+24}, which we will call Algorithm BooleanEF1. For Boolean $\{0,-1\}$-valuations, we use Algorithm NegBooleanEF1.

\begin{restatable}{proposition}{rstbinzeromone}
\label{prop:binZeroMOne2}
For identical $\{0,-1\}$-valuations, the allocation returned by the Algorithm NegBooleanEF1 satisfies one or both of the following conditions.
\begin{enumerate}
    \item Each agent is in $\res^-$ or $\zero$.
 \item The first $n-1$ agents are in $\res^-$.
\end{enumerate}
\end{restatable}

\begin{proposition}
\label{prop:binZeroOne2}
For identical $\{0,1\}$-valuations, the allocation returned by Algorithm BooleanEF1 satisfies one or both of the following conditions:
\begin{enumerate}
    \item Each agent is in $\res^+$ or $\zero$.
 \item The first $n-1$ agents are in $\res^+$.
\end{enumerate}
\end{proposition}

\begin{algorithm}[!ht]
\caption{TrileanNegEF1}
\label{alg:trileanIdent}
\begin{algorithmic}[1]
\REQUIRE Fair division instance $(N,M,\mathcal{V})$  with identical negative trilean valuations.
\ENSURE An \EF{1} allocation $A$.
\STATE Initialize $A=(\emptyset,\hdots,\emptyset)$, $M'=M$, and $i=1$.
\WHILE{($\exists S \subseteq M'$ such that $S$ is favourable) AND ($i < n$)}
	\STATE $A_i = S$, $M' = M' \setminus S$, $i = i+1$ \label{line:negfavourable} \COMMENT{Assign favourable sets.}
\ENDWHILE

\WHILE[Assign flexible sets.]{($M' \neq \emptyset$) AND ($i < n$) AND ($M'$ is trilean) AND ($v(M') \neq 0$)}
      \IF{$v(M')=1$}
            \STATE Let $S$ be an inclusion-wise maximal subset such that $v(S)=-1$.
            \STATE Pick any $x \notin S$. $A_i = S \cup \{x\}$, $M'=M'\setminus A_i$ .   \COMMENT{$v(A_i)=0 \rightarrow -1$.} \label{flex1}
     \ELSE
            \STATE Let $S$ be a inclusion-wise maximal subset such that $v(S)=1$.
            \STATE Pick any $x \notin S$. $A_i = S \cup \{x\}$, $M'=M'\setminus A_i$.  \COMMENT{$v(A_i)=0 \rightarrow 1$.}  \label{flex2}
     \ENDIF
\STATE $i = i+1$
\ENDWHILE
\IF {($M' = \emptyset$)}
	\STATE \textbf{return} allocation $A$. \label{Mempty}
\ENDIF
\IF {($v(M') = 0$)}
	\STATE $A_i = M'$, \textbf{return} allocation $A$. \label{Mzero}
\ENDIF
\IF {($i=n$)}
	\STATE $A_i = M'$, $A = $ FixEF1ViolationsNeg$(A)$, \textbf{return} allocation $A$. \label{lastagent}
\ENDIF
\IF[Assign resolved sets.]{($M'$ is Boolean $\{0,1\}$-valued)} 
           \STATE $A =$ Algorithm BooleanEF1$(M',N \setminus [i-1],\mathcal{V})$. \label{binZeroOneAlloc}
           \STATE $A = $ FixEF1ViolationsNeg$(A)$,  \textbf{return} allocation $A$. \label{fix2}
      \ELSE
          \STATE $A =$ Algorithm~\ref{alg:binZeroMOne}$(M',N \setminus [i-1],\mathcal{V})$. \label{binZeroMOneAlloc}
           \STATE $A = $ FixEF1ViolationsNeg$(A)$,  \textbf{return} allocation $A$. \label{fix3}
      \ENDIF    
\end{algorithmic}
\end{algorithm}

Algorithm TrileanNegEF1 terminates in one of five places --- Lines~\ref{Mempty},~\ref{Mzero},~\ref{lastagent},~\ref{fix2}, and~\ref{fix3}.
The next claim shows that if the algorithm terminates in Line~\ref{Mempty} or Line~\ref{Mzero}, the allocation returned is EF1.

\begin{clm}
If Algorithm TrileanNegEF1 terminates in Line~\ref{Mempty} or Line~\ref{Mzero}, the allocation returned is EF1.
\label{claim:emptyzero}
\end{clm}

\begin{proof}
Note that prior to Lines~\ref{Mempty} and~\ref{Mzero}, any agents with non-empty bundles were assigned either favourable sets or flexible sets in the preceding while loops. In the execution of Lines~\ref{Mempty} or~\ref{Mzero}, every remaining agent is either unassigned or assigned a zero-valued bundle. From Lemma~\ref{lem:ef1violations}, there is no EF1 violation between favourable, flexible, and zero-valued agents. Hence the resulting allocation is EF1.  
\end{proof}

If Line~\ref{lastagent} executes, then before FixEF1ViolationsNeg$(A)$ is called, agents $1, \hdots, n-1$ are assigned either favourable or flexible sets, and hence by Lemma~\ref{lem:ef1violations}, any EF1 violation must involve agent $n$. Claim~\ref{claim:trileanIdentViol} then details the possible EF1 violations before Algorithm FixEF1ViolationsNeg is called. 

For the remaining Lines~\ref{fix2} and~\ref{fix3}, the following proposition states what the allocation looks like before Algorithm FixEF1ViolationsNeg is called.

\begin{clm}
Let $k$ be the last agent to be allocated in the while loop in Algorithm TrileanNegEF1. Then clearly agents $1, \hdots, k$ are either favourable or flexible. Further,
\begin{enumerate}
\item If FixEF1ViolationsNeg$(A)$ is called in Line~\ref{fix2}, then either:
\begin{enumerate} 
\item agents $k+1, \hdots, n$ are in $\res^+ \cup \zero$, or
\item agents $k+1, \hdots, n-1$ are in $\res^+$, and $A_n$ is Boolean $\{0,1\}$-valued.
\end{enumerate}
\item If FixEF1ViolationsNeg$(A)$ is called in Line~\ref{fix3}, then either:
\begin{enumerate} 
\item agents $k+1, \hdots, n$ are in $\res^- \cup \zero$, or
\item agents $k+1, \hdots, n-1$ are in $\res^-$, and $A_n$ is Boolean $\{0,-1\}$-valued.
\end{enumerate}
\end{enumerate}
\label{claim:values}
\end{clm}

\begin{proof}
Prior to calling FixEF1ViolationsNeg$(A)$ in Line~\ref{fix2}, the algorithm calls Algorithm BooleanEF1 with agents $N \setminus [k]$ and items $M'$ that are Boolean $\{0,1\}$-valued. From Proposition~\ref{prop:binZeroOne2}, either each agent $i > k$ is in $\res^+ \cup \zero$, or agents $k+1, \hdots, n-1$ are in $\res^+$, and $A_n$ is Boolean $\{0,1\}$-valued. This proves the claim for Line~\ref{fix2}. A similar proof (using Proposition~\ref{prop:binZeroMOne2}) shows the claim for Line~\ref{fix3}. 
\end{proof}

We then have the following claim, regarding possible EF1 violations in the allocation passed to Algorithm FixEF1ViolationsNeg.

\begin{restatable}{clm}{rsttrileanidentviol}
\label{claim:trileanIdentViol}
Let $A$ be the allocation given as input to Algorithm FixEF1ViolationsNeg. Then any EF1 violation must be of one of the following types:
\begin{itemize}
\item Type 1: Agent $i \in [n-1]$ is in $\flex^-$ and agent $n$ is in $\bad^+$. Other agents are favourable, flexible, or in $\res^+$. 
\item Type 2: Agent $i \in [n-1]$ is in $\flex^+$ and agent $n$ is in $\bad^-$. Other agents are favourable, flexible, or in $\res^-$. 
\end{itemize}
\end{restatable}

Finally, in Algorithm FixEF1ViolationsNeg, we resolve any violations with agent $n$. Claim~\ref{claim:trileanIdentViol} tells us that any EF1 violations must be between agent $n$ and an agent $i$ assigned a flexible set. Further, in this case, $n$ must be a bad agent ($v(A_n) = 1 \rightrightarrows 1$ or $v(A_n) = -1 \rightrightarrows -1$), and $i$ must be a flexible agent of the opposite sign ($v(A_i) = 0 \rightarrow -1$ or $v(A_i) = 0 \rightarrow 1$, respectively). Assume $n \in \bad^+$. Then Algorithm FixEF1ViolationsNeg proceeds by picking an agent $i$ in $\flex^-$ and transferring items (arbitrarily picked) from $A_n$ to $A_i$, until at least one of them is in $\res^+$ (has value $1 \rightarrow 0$). If $n \in \res^+$, we have reached an EF1 allocation. Else, the set of agents in $\flex^-$ is reduced. We then pick the next agent $i$ from $\flex^-$ and continue transferring items from $A_n$ to $A_i$.

We call the \texttt{repeat...until} loops in the algorithm the \emph{inner} loops, and the while loops the \emph{outer} loops. Note that if the initial allocation is not EF1, then from Claim~\ref{claim:trileanIdentViol} either $n \in \bad^+$ or $n\in \bad^-$, and hence exactly one of the two \texttt{if} conditions holds true. 

We claim that when an inner loop terminates, either agent $n$ is resolved and the algorithm terminates with an EF1 allocation, or the chosen agent $i$ is in resolved and $n$ is in \bad. We first show that each inner \texttt{repeat...until} loop runs for at most $|A_n|-1$ iterations over all iterations of the outer while loop.

\begin{algorithm}[ht]
\caption{FixEF1ViolationsNeg}
\label{alg:fixViolations}
\begin{algorithmic}[1]
\REQUIRE An allocation $A$ with possible EF1 violations.
\ENSURE An \EF{1} allocation $A$.
\IF{(Allocation $A$ is \EF{1})}
    \STATE \textbf{return} allocation $A$. \label{line:FixEasy}
\ENDIF
\IF[$v(A_n) = 1 \rightrightarrows 1$ and $\exists i: v(A_i) = 0 \rightarrow -1$]{($n \in \bad^+$) AND ($\flex^- \neq \emptyset$)} 
     \WHILE{($n \in \bad^+$) AND ($\flex^- \neq \emptyset$)} \label{while1}
           \STATE Let $i \in \flex^-$. 
           \REPEAT \label{line:FixRepeat1}
            \STATE Choose an item $x \in A_n$, $A_n = A_n \setminus \{x\}$, $A_i=A_i \cup \{x\}$. 
         \UNTIL{($i \in \res^+$) OR ($n \in \res^+$)} \label{line:FixUntil1} \COMMENT{Either ($v(A_i) = 1 \rightarrow 0$) or ($v(A_n) = 1 \rightarrow 0$).}
     \ENDWHILE
	\STATE \textbf{return} allocation $A$. \label{line:FixPositive}
\ENDIF
\IF[$v(A_n) = -1 \rightrightarrows -1$ and $\exists i: v(A_i) = 0 \rightarrow 1$]{($n \in \bad^-$) AND ($\flex^+ \neq \emptyset$)} 
     \WHILE{($n \in \bad^-$) AND ($\flex^+ \neq \emptyset$)} \label{while2}
           \STATE Let $i \in \flex^+$. 
           \REPEAT \label{line:FixRepeat2}
            \STATE Choose an item $x \in A_n$, $A_n = A_n \setminus \{x\}$, $A_i=A_i \cup \{x\}$. 
         \UNTIL{($i \in \res^-$) OR ($n \in \res^-$)} \label{line:FixUntil2} \COMMENT{Either ($v(A_i) = -1 \rightarrow 0$) or ($v(A_n) = -1 \rightarrow 0$).}
     \ENDWHILE
	\STATE \textbf{return} allocation $A$. \label{line:FixNegative}
\ENDIF
\end{algorithmic}
\end{algorithm}

\begin{restatable}{clm}{rstinnerlooptime}
Let $t = |A_n|$ be the initial size of $A_n$. Each inner loop terminates in at most $t-1$ iterations over all invocations. 
\label{claim:innerlooptime}
\end{restatable}

\begin{proof}
Note that each iteration of an inner loop removes an item from $A_n$. We prove the claim for the first inner loop (Line~\ref{line:FixRepeat1} to~\ref{line:FixUntil1}). The proof for the second inner loop proceeds similarly.

Initially $v(A_n) = 1 \rightrightarrows 1$. Since all favourable sets have already been assigned, there is no set $S \subseteq A_n$ with $v(S) = 1 \rightarrow -1$. Hence if $v(A_n) \neq 1 \rightrightarrows 1$, then it must be that $v(A_n) = 1 \rightarrow 0$. Thus if the first inner loop has not terminated after $t-1$ iterations, then $|A_n| = 1$, $v(A_n) = 1$, and any child of $A_n$ is an empty set of value $0$, and hence the loop terminates. A similar proof holds for the second inner loop as well, when initially $v(A_n) = -1 \rightrightarrows -1$. 
\end{proof}

\begin{restatable}{clm}{rstfix}
After every iteration of the first while loop (Line~\ref{line:FixUntil1}), either agent $n$ moves from $\bad^+$ to $\res^+$ and the algorithm returns an EF1 allocation, or agent $i$ moves from $\flex^-$ to $\res^+$ and agent $n$ remains in $\bad^+$.

Similarly, after every iteration of the second while loop (Line~\ref{line:FixUntil2}), either agent $n$ moves from $\bad^-$ to $\res^-$ and the algorithm returns an EF1 allocation, or agent $i$ moves from $\flex^+$ to $\res^-$ and agent $n$ remains in $\bad^-$.
\label{claim:fix1}
\end{restatable}

\begin{proof}
We show the claim for the first while loop. The case of the second while loop proceeds similarly.

We first show how the values for the agents change as items are transferred in the inner loop. Let $A^0$ be the initial allocation passed to Algorithm FixEF1ViolationsNeg. Note that each flexible agent in $A^0$ is assigned their bundle in the second while loop in Algorithm TrileanNegEF1. Let $i$ be such a flexible agent, and let $M'$ be the set of remaining items from which agent $i$ is allocated. Then $v(A_i^0) = 0\rightarrow -1$, and for any subset of items $S \subseteq M' \setminus A_i^0$, $v(A_i^0 \cup S) \neq -1$. Hence in particular, for $S \subseteq A_n^0$, $v(A_i^0 \cup S) \neq -1$. Thus for any $S \subseteq A_n^0$, either $v(A_i^0 \cup S)=0$, or $v(A_i^0 \cup S) = 1 \rightarrow 0$ ($i$ is moved to $\res^+$ and the inner loop terminates).

Initially, $v(A_n^0) =1 \rightrightarrows 1$. Again, any favourable set has been previously assigned, and hence as we remove items from $A_n$, we must encounter a zero before we encounter a $-1$. Thus for any $S \subseteq A_n^0$, either $v(A_n^0 \setminus S)=1 \rightrightarrows 1$, or $v(A_n^0 \setminus S) = 1 \rightarrow 0$ ($n$ is moved to $\res^+$ and the inner loop terminates).

We now prove the claim by induction on the number of iterations of the outer while loop. After the first while loop, either agent $i$ is moved from $\flex^-$ to $\res^+$, or agent $n$ is moved from $\bad^+$ to $\res^+$. If agent $n$ moves to $\res^+$, there are no more bad agents, and all resolved agents are of the same sign (all are in $\res^+$). By Lemma~\ref{lem:ef1violations}, the allocation is EF1. Else, agent $i$ moves to $\res^+$, and as observed earlier, agent $n$ must remain in $\bad^+$.

Consider the allocation following the $k$th iteration of the outer while loop. By induction, prior to this iteration, $k-1$ agents were moved from $\flex^-$ to $\res^+$, and agent $n$ was in $\bad^+$. By the same argument as in the base case, in the $k$th iteration either the  $n$th agent is no longer a bad agent and the allocation is EF1, or the chosen agent $i$ is moved from $\flex^-$ to $\res^+$.
\end{proof}

We now complete the proof of our main theorem, showing existence of EF1.

\begin{proof}[Proof of Theorem~\ref{theorem:trilean}]
We show that Algorithm TrileanNegEF1 returns an EF1 allocation. By Claim~\ref{claim:emptyzero}, if Algorithm TrileanNegEF1 terminates in Line~\ref{Mempty} or Line~\ref{Mzero}, the allocation returned is EF1. Otherwise, the algorithm calls Algorithm FixEF1ViolationsNeg to fix the allocation. Claim~\ref{claim:trileanIdentViol} then shows that for the allocation passed to Algorithm FixEF1ViolationsNeg, any EF1 violation must be of either Type 1 or Type 2.

Suppose the EF1 violation is of Type 1. The case when the violation is of Type 2 is handled similarly. For a Type 1 violation, there must be agents in $\flex^-$, agent $n$ is in $\bad^+$, and the other agents are in \fav, $\res^+$, or $\flex^+$. Since there is an EF1 violation, and agent $n$ is in $\bad^+$, the first while loop in Algorithm FixEF1ViolationsNeg will execute, selecting an agent $i$ from $\flex^-$. From Claim~\ref{claim:fix1}, each time the inner loop terminates, either agent $n$ moves from $\bad^+$ to $\res^+$ and the algorithm returns an EF1 allocation, or agent $i$ moves from $\flex^-$ to $\res^+$ and agent $n$ remains in $\bad^+$ (and note that by Claim~\ref{claim:innerlooptime}, the inner loop terminates in at most $m$ iterations). Thus, eventually, the algorithm either moves agent $n$ from $\bad^+$ to $\res^+$ and returns an EF1 allocation, or moves all the agents in $\flex^-$ to $\res^+$. At this point, all agents are \fav, $\flex^+$, $\res^+$, or $\bad^+$. By Lemma~\ref{lem:ef1violations}, this is an EF1 allocation. 
\end{proof}

We note that our algorithm may need to make an exponential number of queries, e.g., in the very first step to find favourable sets. Given the lack of structure in the problem, we do not know if there is a way of avoiding this.

\subsection{Positive Trilean Valuations}
\label{sec:trileanIdentPos}

We now turn our attention to the case where agents are identical and their valuations are positive trilean, i.e., $v_i(S) \in \{0,1,2\}$ for all agents $i \in N$ and subsets $S \subseteq M$. Similar to before, we show the existence of EF1 valuations. The algorithm is similar to the case of negative trilean valuations, with some simplifications and modifications due to the positive valuations, and is given in the appendix.

\begin{restatable}{theorem}{rsttrileanpos}
Given an instance with identical positive trilean valuations, Algorithm TrileanPosEF1 returns an EF1 allocation.
\label{theorem:trileanpos}
\end{restatable}


\section{Separable Single-Peaked Valuations}

We now turn to EF1 allocations for separable, single-peaked (SSP) valuations. As before, $N$ and $M$ denote the set of agents and items respectively. The set $M$ is partitioned into $t$ types $(M_1, \ldots, M_t)$, with $m_j = |M_j|$ for $j \in [t]$. For an allocation $A$, we denote agent $i$'s bundle $A_i$ as a $t$-tuple $(a_{i1}, a_{i2}, ... , a_{it})$, where $a_{ij}$ denotes the number of items of type $j$ allocated to agent $i$.

To define the valuation functions for the agents, we first define $\theta_{ij}$ for $i \in N$ and $j \in [t]$ as the threshold of agent $i$ for items of type $j$. Then agent $i$'s valuation $v_i(A_i) = \sum_{j=1}^t v_{ij}(a_{ij})$, where the valuations $v_{ij}$ are single-peaked: for all $x \le y \le \theta_{ij}$, $v_{ij}(x) \le v_{ij}(y)$, while for $\theta_{ij} \le x \le y$, $v_{ij}(x) \ge v_{ij}(y)$.

We point out one basic problem that occurs with SSP valuations, that must be overcome by an algorithm returning EF1 allocations. Consider a simple instance with just two agents and a partial EF1 allocation $A$, where agent $1$ envies agent $2$. Suppose there is an item $x$ that remains to be assigned. Item $x$ is a chore for agent 1 given $A_1$ \emph{and is a good for agent 1 given $A_2$}. That is, $v_1(A_1 \cup \{x\}) < v_1(A_1)$, and $v_1(A_2 \cup \{x\}) > v_2(A_2)$. There is no obvious way to assign item $x$ while maintaining the EF1 property. This situation does not arise with doubly-monotone valuations.

For SSP valuations, we show two results. Firstly, we show existence of EF1 allocations when for each type $j$, all agents have the same threshold $\theta_j$. In this case, our proof shows that the two-phase algorithm for doubly monotone valuations~\cite{BSV21approximate} works in this case as well. We reproduce the algorithm and provide the complete proof in the appendix.

\begin{restatable}{theorem}{rstsspequal}
For SSP valuations where for each type, there is a common threshold for all agents, there always exists an EF1 allocation.
\label{thm:sspequal}
\end{restatable}

We however present an example here to show that while the algorithm is the same as for doubly monotone valuations, the analysis of the algorithm is more complicated. The complications in the example carry over to the three agent case with non-identical thresholds as well.

\paragraph*{An example with identical thresholds}

There are two identical agents, $1$ and $2$ with identical valuations, and two types of items. There is a single item $x$ of type $1$, and $v(x) = 5$. There are three items $a$, $b$, and $c$ of type $2$. Each agent has value $10$ if she gets a single item of type 2, value $9$ if she gets two, and value $-10$ if she gets all three. Thus, the common threshold for the second type is 1.

After phase 1, the allocation is $A_1 = \{a\}$, and $A_2 = \{x,b\}$, and $1$ envies $2$. Note that removing item $b$ removes this envy.

The algorithm in Phase 2 then assigns $c$ to agent $2$. Then $v(A_1) = 10$ and $v(A_2) = 14$. However, note that removing item $b$ --- the item which upon removal, previously eliminated the envy --- from $A_2$ does not eliminate envy in the current allocation. In fact, $v(\{x,c\}) = 15 > v(\{x,b,c\})$. Hence, the following proof by induction used earlier in the doubly monotone case --- that argues that since there existed an item $x$ which upon removal, eliminated envy earlier, the removal of $x$ eliminates envy in the current allocation --- does not work.

The trouble is that within an item type, valuations are not additive. For some allocation $A$, if $v_1(A_1) < v_1(A_2)$ and removing item $x$ of type $j$ from $A_2$ reduced agent $1$'s value for $A_2$ by $\delta$, then since the valuation for further items of type $j$ is not additive, if we assign another item of type $j$ to agent 2 --- \emph{even if this item has negative marginal value for all agents and bundles} --- removing the earlier item $x$ may not reduce the value by $\delta$. In fact, removal of $x$ may even increase the valuation, as seen in the example. Thus while assigning items, we have to be very careful to avoid this situation. In particular, the second subcase of Lemma~\ref{lem:identThresh2} (in the appendix) is significantly more complicated to deal with this scenario.
\qed

Secondly, for the general case where agents may have different thresholds for a type, we show the existence of EF1 allocations for three agents. Here, the algorithm considerably differs from the algorithm for doubly-monotone valuations. While it is still a two-phase algorithm, the two phases are very carefully crafted for the special case of three agents. 

\begin{algorithm}[ht]
\begin{algorithmic}[1]
\STATE Initialise $A_i$ to $(0, 0, ..., 0)$ for each agent $i \in \{1,2,3\}$
\STATE Initialise $\bar{m}_j \leftarrow \lfloor m_j/3 \rfloor$, $\hat{m}_j \leftarrow m_j \bmod 3$, $N_j \leftarrow \{i:\  \theta_{ij} > \bar{m}_j\}$ for each item type $j \in [t]$ \COMMENT{$\hat{m}_j \in \{0,1,2\}$ is the number of items left over after each agent is assigned $\lfloor m_j/3\rfloor$ items, and $N_j$ is the set of agents that still want more items.}
\FORALL[Phase 1] {$j\ \in\ [t]$ s.t. $|N_j| \ge \hat{m}_j$}
    \WHILE {$G_A$ has a cycle $C$}
        \STATE $A \leftarrow A_C$ \label{line:sspass1} \COMMENT{Swap bundles along $C$}
    \ENDWHILE
    \STATE Let $1$, $2$, $3$ be the topological order of the agents in $G_A$
    \STATE $a_{ij} \leftarrow \bar{m}_j$ for each agent $i$ \label{line:sspass2}
    \FOR {$i=1$ to $3$}
        \IF {$i \in N_j$ and $\hat{m}_j > 0$} 
            \STATE $a_{ij} \leftarrow a_{ij} + 1$ \label{line:sspass3} \COMMENT{If $i$ wants an item, and there is an item remaining, assign it}
            \STATE $\hat{m}_j \leftarrow \hat{m}_j - 1$
        \ENDIF
    \ENDFOR
\ENDFOR
\FORALL[Phase 2, $\hat{m}_j \in \{1,2\}$, $|N_j| \in \{0,1\}$] {$j\ \in\ [t] $ s.t. $\hat{m}_j > |N_j|$} 
    \STATE $C \leftarrow $ any cycle in  $T_A$, $A \leftarrow A_C$ \label{line:sspass4} \COMMENT{Swap bundles along top-trading cycle $C$}
    \STATE $a_{ij} \rightarrow \bar{m}_j$ for each agent $i$ \label{line:sspass5}
    \IF {$|N_j| = 0$} 
        \WHILE {$\hat{m}_j > 0$} 
            \STATE Choose a sink $k$ in the graph $G_A$
            \STATE $a_{kj} \leftarrow a_{kj} + 1$, $\hat{m}_j \leftarrow \hat{m}_j - 1$ \label{line:sspass18}
            \STATE $C \leftarrow $ any cycle in  $T_A$, $A \leftarrow A_C$ \label{line:sspass7} \COMMENT{Swap bundles along top-trading cycle $C$}
        \ENDWHILE
    \ELSIF {$|N_j| = 1$}
        \STATE Let $N_j = \{k\}$, and let $\ell$ be a sink in $G_A$.
        \STATE $a_{kj} \leftarrow a_{kj}+1$
        \label{line:sspass8}
        \STATE $a_{\ell j} \leftarrow a_{\ell j} + 1$ \label{line:sspass9}
    \ENDIF
\ENDFOR
\STATE Return $A$
\caption{An EF{1} algorithm for separable single-peaked valuations with 3 agents}
\label{alg:S3EF1}
\end{algorithmic}
\end{algorithm}

\begin{restatable}{theorem}{rstsspthree}
For SSP valuations with 3 agents, there always exists an EF1 allocation.
\label{thm:sspthree}
\end{restatable}

For the algorithm, given a partial allocation $A$, as in prior work, we let $G_A = (V, E)$ denote the envy graph, where $V = N$ and $(i, k)\ \in\ E$ if agent $i$ envies agent $k$. $T_A$ is the top-trading envy graph, a subgraph of $G_A$ with a directed edge $(i,k)$ if $v_i(A_{k})$ =$\max_{i' \in N} v_i(A_{i'})$ $> v_i(A_i)$. Given a directed cycle $C$ in $G_A$ or $T_A$, $A_C$ is the allocation obtained by giving each agent in $C$ the bundle of the agent they envy in $C$.

The proof of Theorem~\ref{thm:sspthree} follows immediately from the next two lemmas. In order to show that the allocation remains EF1 after the allocation of items of type $j$, we use $\hat{A}$ to denote the allocation prior to type $j$, i.e., after allocation of items of the previous type. At a high level, the proof shows that if agent $i$ envies agent $k$ after the allocation of an item of type $j$, then either removing the item just allocated resolves the envy, or that agent $i$ envied agent $k$ in allocation $\hat{A}$ as well. In the latter case, since allocation $\hat{A}$ is EF1 by induction, there is some item $x$ which upon removal resolves the envy. We show that, by additivity of the valuation across types, removing this item $x$ resolves the current envy as well. As demonstrated in the example above, it is crucial that this item $x$ is not an item of type $j$, else we cannot assume additivity.

\begin{lemma}
    Let $A$ be the (partial) allocation at any point in Phase 1 (i.e., in the first for all loop). Then if agent $i$ envies agent $k$, there exists some item $x \in A_{k}$ such that $v_i(A_i) \ge v_{i}(A_{k} \setminus \{x\})$. 
\end{lemma}

\begin{proof}
The proof is by induction. Initially, each $A_i = (0, 0, ..., 0) $ and the statement clearly holds. We will consider each line where the allocation is changed in Phase 1 and show that the statement holds.

\paragraph*{Line~\ref{line:sspass1} (cycle swapping)} The allocation remains EF1 due to the induction hypothesis. As stated, let $\hat{A}$ be the partial allocation before the cycle swapping step. By the induction hypothesis, the statement holds for $\hat{A}$. Let agent $i$ envy agent $k$ in $A$. If $i \notin C$, then the statement trivially holds due to the induction hypothesis (and since bundles only change hands, but do not change, in the cycle swap). If $i \in C$, $i$'s value for her bundle increases, and due to the induction hypothesis, the envy towards $k$'s bundle remains EF1. 

Note that when all cycles are resolved, we have an acyclic graph (DAG) and hence can arrange the vertices (agents) in a topological order.

\paragraph*{Line~\ref{line:sspass2} (assigning $\bar{m}_j$ items to each agent)} Since we are giving an equal number of items of the same type to each agent, and the valuations are additive across different types, for an agent the relative value of any other bundle remains the same. Hence if removing an item $x$ from an envious agent resolved the envy earlier, this continues to be the case.

\paragraph*{Line~\ref{line:sspass3} (adding an item to an agent in $N_j$)} Let $A$ be the allocation after Phase 1, when Line~\ref{line:sspass3} finishes executing. Suppose agent $i$ envies agent $k$ in $A$. If neither $i$ nor $k$ are allocated in Line~\ref{line:sspass3}, then the statement trivially holds from the previous paragraph. We consider at the other three cases separately. In each case, we will show that $i$ envies $k$ in allocation $\hat{A}$ as well, and if removing item $x$ from $\hat{A}_i$ resolved the envy earlier, this continues to hold.

\begin{enumerate}
    \item $i$ is allocated, $k$ is not. In this case, we know that $i \in N_j$ and hence $\theta_{ij} > \bar{m}_j$. Then 
    
    \begin{align*}
    v_i(A_i) = v_i(\hat{A}_i) + v_{ij}(\bar{m}_j + 1) < v_i(A_k) = v_i(\hat{A}_k) + v_{ij}(\bar{m}_j) \, .
    \end{align*}
    
    But since $v_{ij}(\bar{m}_j + 1) > v_{ij}(\bar{m}_j)$, we have $v_i(\hat{A}_i) < v_i(\hat{A}_k)$, and hence $i$ envies $k$ in $\hat{A}$ as well. 
    
    Now by the induction hypothesis, there exists $x \in \hat{A}_i$ such that $v_i(\hat{A}_i \setminus \{x\}) \ge v_i(\hat{A}_k)$. Then clearly, since the valuation is additive across types, 
    
    \begin{align*}
    v_i(A_i \setminus \{x\}) = v_i(\hat{A}_i \setminus \{x\}) + v_{ij}(\bar{m}_j + 1) \ge v_i(\hat{A}_k) + v_{ij}(\bar{m}_j) = v_i(A_k)\, .
    \end{align*}
    
    \item $i$ is not allocated, $k$ is. Since $i$ is not allocated an item, and we allocate items in Line~\ref{line:sspass3} in DAG order, either $i \notin N_j$ or $v_i(\hat{A}_i) \geq v_i(\hat{A}_k)$.
    
    \begin{enumerate}
    \item Suppose $v_i(\hat{A}_i) \geq v_i(\hat{A}_k)$. Then $v_i(A_i) = v_i(\hat{A}_i) + v_{ij}(\bar{m}_j) \geq v_i(\hat{A}_k) + v_{ij}(\bar{m}_j)$, and hence removing the last item of type $j$ given to agent $k$ in Line~\ref{line:sspass3} resolves the envy.
    
        \item Suppose $v_i(\hat{A}_i) < v_i(\hat{A}_k)$, and $i \notin N_j$. Then by the induction hypothesis, there exists $x \in \hat{A}_i$ so that $v_i(\hat{A}_i \setminus \{x\}) \ge v_i(\hat{A}_k)$. Since $i \notin N_j$, $v_{ij}(\bar{m}_j) \ge v_{ij}(\bar{m}_j + 1)$. Then

        \[
        v_i(A_i \setminus \{x\}) = v_i(\hat{A}_i \setminus \{x\}) + v_{ij}(\bar{m}_j) \ge v_i(\hat{A}_k) + v_{ij}(\bar{m}_j + 1) = v_i(A_k).
        \]
    \end{enumerate}
    
    \item Both $i$, $k$ are allocated. Then for $i$, the value of both bundles goes up additively by $v_{ij}(\bar{m}_j + 1)$. By the induction hypothesis, the statement holds.
\end{enumerate}
\end{proof}

\begin{lemma}
    Let $A$ be the (partial) allocation at any point in Phase 2 (i.e., in the second for all loop). Then if agent $i$ envies agent $k$, there exists some item $x \in A_i \cup A_{k}$ such that $v_i(A_i \setminus \{x\}) \ge v_{i}(A_{k} \setminus \{x\})$. 
\end{lemma}

\begin{proof}
    The proof is by induction. Initially, we have the partial allocation from Phase 1 which we proved was EF1 in the previous lemma. We will consider each line where the allocation is changed in Phase 2 and show that the statement holds.

\paragraph*{Lines~\ref{line:sspass4} and~\ref{line:sspass7} (resolving top-trading envy cycles)} Clearly, resolving top-trading envy cycles maintains EF1, since after resolving the cycle (i) agents in the cycle do not envy any other agent, and (ii) for agents outside the cycle, only the ownership of bundles changes, not the bundles themselves.

\paragraph*{Line~\ref{line:sspass5} (assigning $\bar{m}_j$ items to each agent)} The proof is the same as in the previous lemma, with the modification that the item $x$ that upon removal resolves envy does not have to belong to the envious agent.

\paragraph*{Line~\ref{line:sspass18} (assigning items to sinks if $|N_j| = 0$)} Suppose agent $i$ envies agent $k$ in $A$ after an execution of Line~\ref{line:sspass18}. If neither $i$ or $k$ was allocated an item, the allocation is EF1 by induction. If $i$ was just allocated an item, then $i$ was a sink prior to the allocation, and clearly removing the item just allocated to $i$ resolves the envy. Else, suppose $k$ was allocated an item, then $k$ was a sink prior to the allocation. Note that there were either 1 or 2 items of type $j$ remaining after the equipartition, and hence agent $i$ has previously received at most one item of type $j$ in Line~\ref{line:sspass18}. We consider these cases separately.

\begin{enumerate}
    \item Agent $i$ did not previously receive an item of type $j$ in Line~\ref{line:sspass18}. Hence $a_{ij} = \bar{m}_j$. Then let $\hat{A}$ be the allocation prior to the equipartition (i.e., before any items of type $j$ are assigned). Then
    \[
    v_i(A_i) = v_i(\hat{A}_i) + v_{ij}(\bar{m}_j) < v_i(A_k) =  v_i(\hat{A}_k) + v_{ij}(\bar{m}_j + 1) \le v_i(\hat{A}_k) + v_{ij}(\bar{m}_j) \, ,
    \]

    where the last inequality is because $N_j = \emptyset$. Hence, $v_i(\hat{A}_i) < v_i(\hat{A}_k)$, and for some $x \in A_i \cup A_k$, removing $x$ resolves the envy in allocation $\hat{A}$. Then removing $x$ resolves the envy in allocation $A$ as well.
    \item Agent $i$ previously received an item of type $j$ in Line~\ref{line:sspass18}. Hence $a_{ij}=\bar{m}_j+1$, and $a_{kj} = \bar{m}_j+1$. Then
    \[
    v_i(A_i) = v_i(\hat{A}_i) + v_{ij}(\bar{m}_j+1) < v_i(A_k) =  v_i(\hat{A}_k) + v_{ij}(\bar{m}_j + 1) \, .
    \]

    Again, $v_i(\hat{A}_i) < v_i(\hat{A}_k)$, and for some $x \in A_i \cup A_k$, removing $x$ resolves the envy in allocation $\hat{A}$. Then removing $x$ resolves the envy in allocation $A$ as well.
\end{enumerate}

\paragraph*{Line~\ref{line:sspass8} (giving an item to $k \in N_j$)} Since only agent $k$ wants more items after the equipartition (since $|N_j| = 1$), clearly any envy after allocating an item in Line~\ref{line:sspass8}  must have been present in allocation $\hat{A}$ as well. It follows along similar lines to previous proofs, that if removing an item $x$ resolved envy in allocation $\hat{A}$, this is true of the current allocation as well.

\paragraph*{Line~\ref{line:sspass9} (giving the last item to a sink)} As before, let $\hat{A}$ be the allocation before allocating any items of type $j$, and $A$ be the allocation after this line. If agent $\ell$ envies any agent after this line, then clearly removing this last item resolves the envy. Now suppose agent $i$ envies agent $\ell$. If $i \in N_j$, then $i$ also received an item of type $j$ in Line~\ref{line:sspass8}, else $i \notin N_j$. In either case, we can show that $v_i(\hat{A}_i) < v_i(\hat{A}_{\ell})$. As previously, removing the item $x$ that resolves the envy in allocation $\hat{A}$ works for the current allocation as well.
\end{proof}



\section{Non-existence of \efx~Allocations}
\label{sec:non-existEFX}

Given the existence of EF1 allocations, a natural question is whether EFX allocations exist for the valuations studied. B{\'e}rczi et al. refine the definition of EFX for nonmonotone valuations. They define an $\efx^{+}_{-}$ allocation $A$ as one where if agent $i$ envies agent $j$, then this should be resolved by removing any item with a strictly positive marginal value from $A_j$, and any item with a strictly negative marginal value from $A_i$. Further at least one such item must exist. They show that for identical negative Boolean valuations, as well as for positive Boolean valuations, an $\efx^{+}_{-}$ allocation always exists.

More formally, for agent $i$ and subset $S \subseteq M$, let  $M^{+}_{i}(S)$ $=\{x \in S ~|~ v_i(S) > v_i(S \setminus \{x\}) \}$, and $M^{-}_{i}(S)$ $=\{x \in S ~|~ v_i(S) < v_i(S \setminus \{x\})\}$. Then allocation $A$ is $\efx^{+}_{-}$ if, whenever agent $i$ envies agent $j$, the set $M^+_i(A_j) \cup M^-_i(A_i)$ is nonempty, and for every $x \in M^+_i(A_j) \cup M^-_i(A_i)$, $v_i(A_i \setminus \{x\}) \ge v_i(A_j \setminus \{x\})$. To complete the picture, we show that $\efx^{+}_{-}$ allocations do not exist for identical negative trilean valuations, or for separable single-peaked valutions, even with two identical agents and three items of a single type. In fact the same example shows nonexistence for both valuation classes.

\begin{restatable}{theorem}{nonexistefx}
$\efx^{+}_{-}$ allocations may not exist even for two agents with identical negative trilean valuations, or two agents with identical SSP valuations.
\label{theorem:non-existEFXNegTer}
\end{restatable}

\begin{proof}
Consider an instance with three items and two identical agents. For $S \subseteq M$, $v(S) = 0$ if $S = \emptyset$, $1$ if $|S| = 1$, and $-1$ if $|S| \ge 2$. This valuation is clearly both negative trilean and separable single-peaked with a single type. 

To show non-existence of $\efx^{+}_{-}$ there are two cases to consider: (i) agent $1$ gets nothing, and (ii) agent $1$ gets one item. In case (i), since $v(A_1)=0$ and $v(A_2)=-1$, agent $2$ envies agent $1$. However $M^{+}_{2}(A_1) \cup M^{-}_{2}(A_2) = \emptyset$, and hence this allocation is not $\efx^{+}_{-}$. In case (ii), $v(A_1) = 1$ and $v(A_2) = -1$ and again agent $2$ envies agent $1$. However removing the single item from $A_1$ does not remove the envy, and hence this allocation is also not $\efx^{+}_{-}$.
\end{proof}

\section{Conclusion} Our paper extends work on the existence of EF1 allocations in two new directions ---  for trilean valuations, and for separable single-peaked valuations. Both of these are natural classes of valuations in which to study EF1 allocations. For trilean valuations, it appears likely that our algorithm and the structures introduced may be useful in further extending results on the existence of EF1 to general identical valuations. For example, if $\max$ and $\min$ are the maximum and minimum possible values for any set, then a set $S$ with value $v(S) = \max \rightarrow \min$ or $v(S) = \min \rightarrow \max$ should be assigned immediately, similar to how we treat favourable sets. More directly, just as our algorithm uses algorithms for Boolean $\{0,1\}$ and $\{0,-1\}$ valuations as subroutines, it is possible that algorithms for $k$ distinct valuations use algorithms for $(k-1)$ distinct valuations as subroutines.

The big open question that remains is the existence of EF1 allocations, even for three agents with identical valuations. More immediately, EF1 existence is left open for agents with nonidentical trilean valuations, and for more than three agents with separable single-peaked valuations. 



\bibliographystyle{plainurl}
\bibliography{ef1-allocations}

\newpage

\appendix

\section{Appendix}

\subsection{Trilean Valuations}

\rstab*

\begin{proof}
As stated, since for envy-freeness we only compare relative values of bundles, the proof is immediate if either $a$ or $b$ is nonnegative. Thus we only need to consider the case where both $a$ and $b$ are negative. For this, suppose $0 > a > b$. Create a trilean instance $(N,M,\mathcal{V}')$ where we set $v'(S) = 1$ if $v(S) = a$, $v'(S) = 2$ if $v(S) = b$, and $v'(S) = 0$ if $v(S) = 0$. Note that if $v(S) > v(T)$ for $S, T \subseteq M$, then $v'(S) < v'(T)$. Let $A$ be an EF1 allocation in $(N,M,\mathcal{V}')$. We will show that $A$ is an EF1 allocation for the original instance as well. Suppose that $i$ envies $j$, thus $v(A_i) < v(A_j)$. Then $v'(A_i) > v'(A_j)$, and since allocation $A$ is EF1 for valuations $\mathcal{V}'$, there exists some item $x \in A_i \cup A_j$ so that $v'(A_i \setminus \{x\}) \le v'(A_j \setminus \{x\})$. It follows immediately that if $v(A_i) < v(A_j)$, for some item $x \in A_i \cup A_j$,  $v(A_i \setminus \{x\}) \ge v(A_j \setminus \{x\})$, resolving the envy.
\end{proof}


\subsection{Negative Trilean Valuations}

We present proofs of results missing from the main paper.

\rstefviolations*

\begin{proof}
For Figure~\ref{fig:ef1violations}, we will show that if two sets don't have an edge between them, there cannot be an EF1 violation for agents in those sets.

\paragraph*{$\fav$ has no edge.} ($-1 \rightarrow 1$ or $1 \rightarrow -1$) Consider $i \in \fav$ so that $v(A_i) = 1 \rightarrow -1$. Then $i$ does not envy any agent, and any envy towards $i$ can be eliminated by removing an element $x$ so that $v(A_i \setminus \{x\}) = -1$. Hence $i$ is mutually EF1 with every other agent. Similarly, if $v(A_i) = -1 \rightarrow 1$, then no agent envies $i$, and $i$'s envy can be eliminated by removing an item $x$ from $A_i$ so that $v(A_i \setminus \{x\}) = 1$. Hence again, $i$ is mutually EF1 with every other agent.

\paragraph*{$\flex^+$, $\res^+$, $\bad^+$ form an independent set.}  ($0 \rightarrow 1$, $1 \rightarrow 0$, $1 \rightrightarrows 1$) Consider agents $i$, $j$ in $\flex^+ \cup \res^+ \cup \bad^+$. If neither $i$ nor $j$ is in $\flex^+$, then $i$ and $j$ have the same value and are mutually envy-free. If $i \in \flex^+$ and $j \in \res^+ \cup \bad^+$, then $i$ envies $j$, but there exists $x \in A_i$ so that $v(A_i \setminus \{x\}) = 1$.

\paragraph*{$\flex^-$, $\res^-$, $\bad^-$ form an independent set.} ($0 \rightarrow -1$, $-1 \rightarrow 0$, $-1 \rightrightarrows -1$) The proof is similar to the previous case.

\paragraph*{$\flex^-$, $\res^+$, $\zero$ form an independent set.} ($0 \rightarrow -1$, $1 \rightarrow 0$, $0$) Consider agents $i$, $j$ in $\flex^- \cup \res^+ \cup \zero$. As in the previous case, agents in $\flex^-$ and $\zero$ have the same value $0$ for their bundles, and are mutually envy-free. If $i \in \res^+$ and $j \in \flex^- \cup \zero$, then $j$ envies $i$, but since $v(A_i) = 1 \rightarrow 0$ this envy is resolved by removing an item from $A_i$.

\paragraph*{$\flex^+$, $\res^-$, $\zero$ form an independent set.} ($0 \rightarrow 1$, $-1 \rightarrow 0$, $0$) The proof is similar to the previous case.

\paragraph*{$\flex^-$, $\flex^+$, $\zero$ form an independent set.} ($0 \rightarrow -1$, $0 \rightarrow 1$, $0$) This is because agents in any of these sets have value $0$ and are hence mutually envy-free.

This completes the proof.  
\end{proof}

\rstbinzeromone*

\begin{proof}
If Algorithm~\ref{alg:binZeroMOne} is given identical Boolean $\{0,-1\}$ valuations, then in the resulting allocation, every agent $i$ allocated in the while loop has $v(A_i) = -1 \rightarrow 0$, hence $i \in \res^-$. If Line~\ref{line:negativebooleanif} executes, or $M' = \emptyset$ after the while loop, then every remaining agent has value $0$, satisfying the first condition. Else, the first $n-1$ agents have value $v(A_i) = -1 \rightarrow 0$ for their bundles, satisfying the second condition.  
\end{proof}

\rsttrileanidentviol*

\begin{proof}
Consider the 3 places where Algorithm TrileanNegEF1 calls Algorithm FixEF1ViolationsNeg.
\begin{itemize}
\item In Line~\ref{lastagent}, before FixEF1ViolationsNeg$(A)$ is called, agents $1, \hdots, n-1$ are assigned either favourable or flexible sets. From Lemma~\ref{lem:ef1violations}, a flexible set can only have an EF1 violation with a bad set. Hence if there is an EF1 violation, it must be between either $i \in \flex^+$ and $n \in \bad^-$, or $i \in \flex^-$ and $n \in \bad^+$, as claimed.
\item Claim~\ref{claim:values} describes the allocation if FixEF1ViolationsNeg is called in Line~\ref{fix2}. In Case 1(a), all agents are either favourable, flexible, in $\res^+$, or zero-valued. From Lemma~\ref{lem:ef1violations}, in this case, there is no EF1 violation to resolve. In Case 1(b), all agents except the last are either favourable, flexible, or in $\res^+$, and the last agent has a set that is Boolean $\{0,1\}$-valued. Again from Lemma~\ref{lem:ef1violations}, EF1 violations are only possible if $n \in \bad^+$, and some agents are in $\flex^-$.
\item If  FixEF1ViolationsNeg is called in Line~\ref{fix3}, similar to Line~\ref{fix2}, we  use Claim~\ref{claim:values} and Lemma~\ref{lem:ef1violations} to show that the only EF1 violations possible are of Type 2. 
\end{itemize} 
\end{proof}

\subsection{Positive Trilean Valuations}

We prove the following theorem in this section, showing existence of EF1 for positive trilean valuations.

\rsttrileanpos*

For the proof, we first define sets similar to the case of negative trilean valuations, with minor adjustments to certain definitions.

\begin{enumerate}
\item Unallocated: $\unalloc = \{i : A_i = \emptyset\}$.
\item Zero: $\zero = \{i : v(A_i) = 0\}$.
\item Favourable: $\fav = \{i : v(A_i) = 2 \rightarrow 0 \text{ or } v(A_i) = 0 \rightarrow 2\}$.
\item Flexible: $\flex = \{i: v(A_i) = 1 \rightarrow 0\}$.
\item Resolved: $\res =\{i : v(A_i) = 2 \rightarrow 1\} \text{ and } \res^*=\{i: v(A_i) = 1 \rightrightarrows \{1,2\}\}$.
\item Bad: $\bad = \{i: v(A_i) = 2 \rightrightarrows 2\}$.
\end{enumerate}
It is easy to verify that the above sets are not mutually exclusive, but are exhaustive. As before, we also apply these terms to describe the corresponding sets. Therefore, a set of items $S$ is:
\begin{enumerate}
\item Zero-valued if $v(S) = 0$.
\item Favourable if $v(S) = 2 \rightarrow 0$ or $v(S) = 0 \rightarrow 2$.
\item Flexible if $v(S) = 1 \rightarrow 0$.
\item Resolved if $v(S) = 2 \rightarrow 1$ or $v(S) = 1 \rightrightarrows \{1,2\}$.
\item Bad if $v(S) = 2 \rightrightarrows 2$.
\end{enumerate}

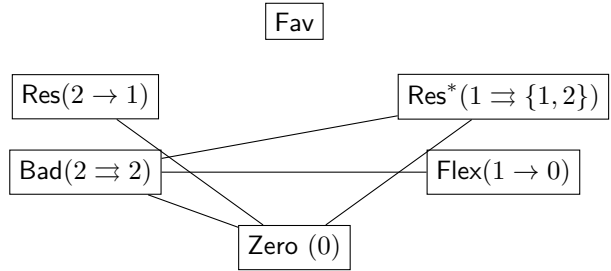
\begin{figure}[!ht]
\vspace{0.1in}
\centering
\begin{tikzpicture}

    \node[draw, rectangle] (X) at (2.75, -2) {$\fav$};
    \node[draw, rectangle] (L1) at (0, -3) {$\res(2 \rightarrow 1)$};
    \node[draw, rectangle] (L2) at (0, -4) {$\bad(2 \rightrightarrows 2)$};
    \node[draw, rectangle] (R1) at (5.5, -3) {$\res^* (1 \rightrightarrows \{1,2\})$};
    \node[draw, rectangle] (R2) at (5.5, -4) {$\flex(1 \rightarrow 0)$};
    \node[draw, rectangle] (Y) at (2.75, -5) {$\zero$ $(0)$};

    \foreach \i in {1,2}
         \draw (L2) -- (R\i);
	\draw (L2) -- (Y);
         \draw (Y) -- (L1);
	\draw (Y) -- (R1);
\end{tikzpicture}
\caption{The edges show possible EF1 violations between different sets of agents}
\label{fig:ef1violationspos}
\end{figure}

It is easy to see that the agents within a set are mutually EF1. As earlier, we now use these sets to give necessary conditions for the violation of EF1.

\begin{lemma}
Given agents with identical positive trilean valuations and an allocation $A$, agents $i$, $j$ are not mutually EF1 only if:\footnote{Note that $\unalloc \subseteq \zero$.}
\begin{enumerate}
\item $i \in \bad$, and $j \in \zero \cup \flex \cup \res^*$, or
\item $i \in \zero$, and $j \in \res \cup \res^*$.
\end{enumerate}
\label{lem:ef1violationspos}
\end{lemma}

Note that agents in set $\fav$ are mutually EF1 with every other agent, irrespective of the bundle possessed by the other agent.

\begin{proof}
This lemma can be easily proved following the proof of Lemma \ref{lem:ef1violations}, and its statement is also conclusive from Figure \ref{fig:ef1violationspos}. Therefore, we will skip the proof.
\end{proof}

We are now ready to describe our algorithm, which has a similar flavor to Algorithm \ref{alg:trileanIdent}. The algorithm operates by initially allocating \emph{favourable} sets to the first $n-1$ agents, as long as such sets are available. If there are items left but no more favourable sets, and not all agents in $n-1$ have been allocated, let $M'$ denote the set of remaining items. If $v(M')=0$ (or $M'=\emptyset$), we assign $M'$ to the next agent and return the resulting allocation. Otherwise, if $M'$ is trilean, then we find a maximal set $S$ with $v(S)=0$ and assign $S \cup \{x\}$ to the next agent, where $x \in M' \setminus S$. Note that $v(S \cup \{x\})=1 \rightarrow 0$ because $M'$ does not contain any favourable sets.

After allocating any flexible sets to agents in $[n-1]$, if at least two agents remain and $v(M')=0$, then assign $M'$ to the next agent and return the resulting allocation, which satisfies EF1. If $v(M') \neq 0$ then $M'$ is Boolean $\{0,1\}$-valued. We then invoke the  algorithm that returns an EF1 allocation for Boolean $\{0,1\}$ valuations. Note that there won't be any EF1 violations in the complete allocation.

The only case that remains is when all agents in $[n-1]$ have been allocated either favourable or flexible sets. In this case, we assign the remaining set of items to agent $n$ and invoke Algorithm \ref{alg:fixViolationsPos} to resolve EF1 violations if there are any.

Note that the set $\res^+$, defined in the context of negative trilean valuations, is the same as $\flex$ in the current context. Hence, Proposition \ref{prop:binZeroOne2} is restated as follows.
\begin{proposition}
\label{prop:binZeroOne2pos}
For identical $\{0,1\}$-valuations, the allocation returned by Algorithm BooleanEF1 satisfies one or both of the following conditions:
\begin{enumerate}
 \item Each agent is in $\flex$ or $\zero$.
 \item The first $n-1$ agents are in $\flex$.
\end{enumerate}
\end{proposition}

\begin{algorithm}[!ht]
\caption{TrileanPosEF1}
\label{alg:trileanPosIdent}
\begin{algorithmic}[1]
\REQUIRE Fair division instance $(N,M,\mathcal{V})$  with identical positive trilean valuations.
\ENSURE  An \EF{1} allocation $A$.
\STATE Initialize $A=(\emptyset,\hdots,\emptyset)$, $M'=M$, and $i=1$.

\WHILE{($\exists S \subseteq M'$ such that $S$ is favourable) AND ($i < n$)}
	\STATE $A_i = S$, $M' = M' \setminus S$, $i = i+1$ \label{line:posfavourable} \COMMENT{Assign favourable sets.}
\ENDWHILE

\WHILE[Assign flexible sets.]{($M' \neq \emptyset$) AND ($i < n$) AND ($M'$ is trilean) AND ($v(M') \neq 0$)} 
\STATE Let $S$ be an inclusion-wise maximal subset such that $v(S)=0$.
\STATE Pick any $x \notin S$. $A_i = S \cup \{x\}$, $M'=M'\setminus A_i$. \COMMENT{$v(A_i)=1 \rightarrow 0$.} \label{flex}
\STATE $i = i+1$
\ENDWHILE

\IF {($M' = \emptyset$)}
	\STATE \textbf{return} allocation $A$. \label{Memptypos}
\ENDIF
\IF {($v(M') = 0$)}
	\STATE $A_i = M'$, \textbf{return} allocation $A$. \label{Mzeropos}
\ENDIF
\IF {($i==n$)}
	\STATE $A_i = M'$, $A = $ FixEF1ViolationsPos$(A)$, \textbf{return} allocation $A$. \label{lastagentpos}
\ENDIF

\IF{($M'$ is Boolean $\{0,1\}$-valued)}
\STATE $A =$ Algorithm BooleanEF1$(M',N \setminus [i-1],\mathcal{V})$, \textbf{return} allocation $A$. \label{binZeroOnepos}
\ENDIF    
\end{algorithmic}
\end{algorithm}

Algorithm TrileanPosEF1 terminates in one of four places --- Lines \ref{Memptypos}, \ref{Mzeropos}, \ref{lastagentpos} and \ref{binZeroOnepos}. The subsequent claim shows that if the algorithm terminates at Line \ref{Memptypos}, Line \ref{Mzeropos}, or Line \ref{binZeroOnepos}, the allocation returned is EF1.

\begin{clm}
If Algorithm TrileanPosEF1 terminates at Line \ref{Memptypos}, Line \ref{Mzeropos}, or Line \ref{binZeroOnepos}, the allocation returned is EF1.
\label{claim:emptyzerobinary}
\end{clm}

\begin{proof}
Prior to Lines~\ref{Memptypos} and~\ref{Mzeropos}, agents with non-empty bundles had already been allocated either favourable sets or flexible sets in the preceding while loops. In the execution of  Lines~\ref{Memptypos} or~\ref{Mzeropos}, every remaining agent is either unassigned or assigned a zero-valued bundle. From Lemma~\ref{lem:ef1violationspos}, there is no EF1 violation between favourable, flexible, and zero-valued sets. Hence, the resulting allocation is EF1.

Let us consider the case in which the algorithm terminates at Line \ref{binZeroOnepos}. Prior to the execution of this line, suppose $k$ is the last agent to be allocated in the while loop. Then agents $1,\hdots,k$ are either favourable or flexible. After Line \ref{binZeroOnepos} executes, according to Proposition \ref{prop:binZeroOne2pos}, either agents $k+1,\hdots,n$ are in $\flex \cup \zero$, or agents $k+1,\hdots,n-1$ are in $\flex$, and $A_n$ is Boolean-valued ($v(A_n)=0$ or $v(A_n)=1$). Using Lemma \ref{lem:ef1violationspos}, we can conclude that the allocation returned in this case is EF1. 
\end{proof}

If Line \ref{lastagentpos} executes, then prior to calling FixEF1ViolationsPos$(A)$, agents $1,\hdots,n-1$ are assigned either favourable or flexible bundles. Thus, by Lemma \ref{lem:ef1violationspos}, any EF1 violation necessarily involves agent $n$. The following claim describes the possible EF1 violations before the FixEF1ViolationsPos algorithm is invoked.

\begin{clm}
\label{claim:trileanIdentViolPos}
Let $A$ be the allocation given as input to Algorithm \ref{alg:fixViolationsPos}. Then any EF1 violation must be of the following type:
\begin{itemize}
\item Agent $i \in [n-1]$ is in $\flex$ and agent $n$ is in $\bad$. All other agents are favourable or flexible. 
\end{itemize}
\end{clm}

\begin{proof}
At Line \ref{lastagentpos}, before calling FixEF1ViolationsPos$(A)$, agents $1,\hdots,n-1$ are assigned either favourabe or flexible bundles. According to Lemma \ref{lem:ef1violationspos}, a flexible bundle can ony have an EF1 violation with a bad bundle. Therefore, if any EF1 violation is present, it must occur between an agent $i \in \flex$ and agent $n \in \bad$.
\end{proof}

Algorithm \ref{alg:fixViolationsPos} then resolves EF1 violations that may exist. According to Claim \ref{claim:trileanIdentViolPos}, such violations involve agent $n$ with a bad bundle ($v(A_n)= 2 \rightrightarrows 2$) and another agent $i$ with a flexible bundle ($v(A_i)= 1 \rightarrow 0$). To resolve EF1 violations, the algorithm picks an agent $i \in \flex$ and transfers items from $A_n$ to $A_i$, until at least one of them is in $\res$ (has value $2 \rightarrow 1$). If agent $n$ is in $\res$, then we have obtained an EF1 allocation. Otherwise, the number of agents in $\flex$ is decreased. The algorithm then picks an agent $i \in \flex$ and continues transferring items from $A_n$ to $A_i$.

As earlier, we call the \texttt{repeat...until} loop in the algorithm the \emph{inner} loop, and the while loop the \emph{outer} loop. We claim that upon termination of an inner loop, either agent $n$ is in $\res$ ($v(A_n)=2 \rightarrow 1$) and the algorithm concludes with an EF1 allocation, or the selected agent $i$ is in $\res$ while agent $n$ remains in $\bad$. Our next claim shows that the inner \texttt{repeat...until} loop runs for at most $|A_n|-2$ iterations over all iterations of the outer while loop.

\begin{algorithm}[!ht]
\caption{FixEF1ViolationsPos}
\label{alg:fixViolationsPos}
\begin{algorithmic}[1]
\REQUIRE An allocation $A$ with possible EF1 violations.
\ENSURE An \EF{1} allocation $A$.
\IF{(Allocation $A$ is \EF{1})}
    \STATE \textbf{return} allocation $A$. \label{line:FixEasypos}
\ENDIF
\IF[$v(A_n) = 2 \rightrightarrows 2$ and $\exists i: v(A_i) = 1 \rightarrow 0$]{($n \in \bad$) AND ($\flex \neq \emptyset$)} 
     \WHILE{($n \in \bad$) AND ($\flex \neq \emptyset$)} \label{while1pos}
           \STATE Let $i \in \flex$. 
           \REPEAT \label{line:FixRepeat1pos}
            \STATE Choose an item $x \in A_n$, $A_n = A_n \setminus \{x\}$, $A_i=A_i \cup \{x\}$. 
         \UNTIL{($i \in \res$) OR ($n \in \res$)} \label{line:FixUntil1pos} \COMMENT{Either ($v(A_i) = 2 \rightarrow 1$) or ($v(A_n) = 2 \rightarrow 1$).}
     \ENDWHILE
	\STATE \textbf{return} allocation $A$. \label{line:FixPositivepos}
\ENDIF
\end{algorithmic}
\end{algorithm}

\begin{clm}
Let $t = |A_n|$ be the size of $A_n$ initially. The inner loop terminates in at most $t-2$ iterations over all invocations. 
\label{claim:innerlooptimepos}
\end{clm}

\begin{proof}
In every iteration of an inner loop, an item is removed from $A_n$ and transferred to $A_i$, where $A_i$ is the bundle of an agent $i$ selected from $\flex$. 
At the beginining, we have $v(A_n) = 2 \rightrightarrows 2$. Since all favourable sets have already been allocated, there is no set $S \subseteq A_n$ with $v(S) = 2 \rightarrow 0$. This also implies that there is no set $S \subseteq A_n$ with $|S|=1$ and $v(S)=2$. Note that if $v(A_n) \neq 2 \rightrightarrows 2$, then it must be that $v(A_n) = 2 \rightarrow 1$. Thus if the inner loop has not terminated after $t-2$ iterations, then $|A_n| = 2$, $v(A_n) = 2$, and after removing any child $x$ of $A_n$, we have $v(A_n \setminus \{x\})=1$, and hence the loop terminates. 
\end{proof}

\begin{clm}
After every iteration of the while loop (Line~\ref{line:FixUntil1pos}) in Algorithm \ref{alg:fixViolationsPos}, either agent $n$ moves from $\bad$ to $\res$ and the algorithm returns an EF1 allocation, or agent $i$ moves from $\flex$ to $\res$ and agent $n$ remains in $\bad$.
\label{claim:fix1pos}
\end{clm}

\begin{proof}
The proof of this claim will closely follow the proof of Claim \ref{claim:fix1}.
Before proving the claim, we show how the values for the agents change as items are transferred in the inner loop. Let $A^0$ be the initial allocation passed to Algorithm FixEF1ViolationsPos. In this allocation, each flexible agent is assigned their bundle in the second while loop in Algorithm \ref{alg:trileanPosIdent}. Let $i$ be such a flexible agent, and let $M'$ be the set of remaining items from which agent $i$ is allocated. Since agent $i$ receives a maximal set such that $v(A_i^0) = 1\rightarrow 0$, for any subset of items $S \subseteq M' \setminus A_i^0$, $v(A_i^0 \cup S) \neq 0$. Specifically, for $S \subseteq A_n^0$, $v(A_i^0 \cup S) \neq 0$. Therefore, for any $S \subseteq A_n^0$, either $v(A_i^0 \cup S)=1$, or $v(A_i^0 \cup S) = 2 \rightarrow 1$ ($i$ is moved to $\res$ and the inner loop terminates).

In $A^0$, agent $n$ has bundle value $v(A_n^0) =2 \rightrightarrows 2$. Since favourable sets have already been allocated, as we remove items from $A_n$, we must encounter a $1$ before we encounter a $0$. Thus for any $S \subseteq A_n^0$, either $v(A_n^0 \setminus S)=2 \rightrightarrows 2$, or $v(A_n^0 \setminus S) = 2 \rightarrow 1$ ($n$ is moved to $\res$ and the inner loop terminates).

Similar to the proof of Claim \ref{claim:fix1}, this claim can be proven using induction on the number of iterations of the outer while loop.
\end{proof}

\begin{proof}[Proof of Theorem~\ref{theorem:trileanpos}.]
We show that Algorithm TrileanPosEF1 returns the necessary allocation. If the algorithm terminates at Line \ref{Memptypos}, Line \ref{Mzeropos}, or Line \ref{binZeroOnepos}, then by Claim \ref{claim:emptyzerobinary}, the allocation returned is EF1.

Consider the case in which the algorithm executes Line \ref{lastagentpos}. After assigning the remaining set of items to the last agent, EF1 violations could possibly occur. By Claim \ref{claim:trileanIdentViolPos}, any such violation must involve agent $n$ with $n \in \bad$ and another agent $i \in \flex$. To resolve such violations Algorithm FixEF1ViolationsPos is invoked. In the while loop (at Line \ref{while1pos}), an agent $i$ from $\flex$ is selected. By Claim \ref{claim:fix1pos}, each time the inner loop terminates, either agent $n$ moves from $\bad$ to $\res$ and the algorithm returns an EF1 allocation, or agent $i$ moves from $\flex$ to $\res$ and agent $n$ continues to remain in $\bad$. Therefore, at the termination of the algorithm either agent $n$ moves from $\bad$ to $\res$, in which case we obtain an EF1 allocation, or all agents in $\flex$ move to $\res$. In that case, all agents are $\fav$, $\res$, or $\bad$. By Lemma \ref{lem:ef1violationspos}, this allocation is EF1.
 
\end{proof}

\subsection{Separable Single-Peaked Valuations}

We present the missing proofs for separable single-peaked valuations.

\paragraph*{Agents with Identical Thresholds}

In this case, for each type $j$, there is a common threshold $\theta_j$ for each agent. Then agent $i$'s valuation $v_i(A_i) = \sum_{j=1}^t v_{ij}(a_{ij})$, where the valuations $v_{ij}$ are single-peaked: for all $x \le y \le \theta_{j}$, $v_{ij}(x) \le v_{ij}(y)$, while for $\theta_{j} \le x \le y$, $v_{ij}(x) \ge v_{ij}(y)$.

For the algorithm, given a partial allocation $A$, as in prior work, we let $G_A = (V, E)$ denote the envy graph, where $V = N$ and $(i, k)\ \in\ E$ if agent $i$ envies agent $k$. $T_A$ is the top-trading envy graph, a subgraph of $G_A$ with a directed edge $(i,k)$ if $v_i(A_{k})$ =$\max_{i' \in N} v_i(A_{i'})$ $> v_i(A_i)$. Given a directed cycle $C$ in $G_A$ or $T_A$, $A_C$ is the allocation obtained by giving each agent in $C$ the bundle of the agent they envy in $C$.

\begin{algorithm}[ht]
\begin{algorithmic}[1]
\STATE Initialise $A_i$ to $(0, 0, ..., 0)\ \forall\ i = 1, 2, ..., n$
\STATE Initialise $\hat{m}_j \leftarrow m_j \ \forall\ j \in [t]$
\FORALL[Phase 1] {$j\ \in\ [t]$}
    \WHILE {$\hat{m}_j > 0$}
        \WHILE {$G_A$ has a cycle $C$}
        \STATE $A \leftarrow A_C$ \label{line:idsspass1} \COMMENT{Swap bundles along $C$}
    \ENDWHILE
        \STATE $N_j = \{i \in N : a_{ij} < \theta_j\}$ \label{line:idsspass2}
        \IF{$|N_j| > 0$}
            \STATE Let $G_A^j$ be the envy graph restricted to $N_j$
            \STATE Choose a source in $G_A^j$, say $p$
            \STATE $a_{pj} \leftarrow a_{pj} + 1$
            \STATE $\hat{m}_j \leftarrow \hat{m}_j - 1$
        \ELSE
            \STATE Break
        \ENDIF
    \ENDWHILE
\ENDFOR
\FORALL[Phase 2] {$j\ \in\ [t]$} 
    \WHILE {$\hat{m}_j > 0$} 
        \WHILE {$T_A$ has a cycle $C$}
        \STATE $A \leftarrow A_C$ \label{line:idsspass3} \COMMENT{Swap bundles along $C$}
    \ENDWHILE    
        \STATE Choose a sink $k$ in the graph $G_A$
        \STATE $a_{kj} \leftarrow a_{kj} + 1$
        \STATE $\hat{m}_j \leftarrow \hat{m}_j - 1$ \label{line:idsspass4}
    \ENDWHILE 
\ENDFOR
\STATE Return $A$
\caption{An EF{1} algorithm for SSP valuations with identical thresholds}
\label{alg:SEF1}
\end{algorithmic}
\end{algorithm}

\rstsspequal*


The proof of Theorem~\ref{thm:sspequal} follows from Lemma \ref{lem:identThresh1} and Lemma \ref{lem:identThresh2}. We note that while the algorithm is as for doubly monotone valuations, the analysis is much more subtle, as shown in the example earlier.

\begin{lemma}
    Let $A$ be the (partial) allocation at any point in Phase 1 (i.e., in the first for all loop). Then if agent $r$ envies agent $s$, then the envy can be resolved by removal of a single item from agent $s$'s bundle. 
\label{lem:identThresh1}
\end{lemma}

\begin{proof}
    The proof is by induction. Initially, each $A_i = (0, 0, ..., 0)$ and the statement clearly holds. We consider the steps in Phase 1 that change the allocation and show that the statement holds. Note also that the number of items of type $j$ in a bundle never exceeds $\theta_j$, since for an agent $i$, if $a_{ij} = \theta_j$, then agent $i \not \in N_j$, and $i$ does not get another item of type $j$. 
    
    Now suppose agent $r$ envies agent $s$ in $A$. The allocations change either when an envy cycle is resolved, or when an item is assigned. If an envy cycle $C$ is resolved, one can verify that the property in the lemma still holds, since if $r$ is not in $C$, from her perspective the allocation does not change; while if $r \in C$, her value goes up, while the value of other bundles remains unchanged. 
    
Suppose an item of type $j$ is assigned to agent $p$. Let $A'$ be the allocation just prior to the allocation, and let $\hat{A}$ be the allocation after the assignment of type $j-1$ in Phase 1. If $r = p$, her valuation simply increased in $A$ wrt $A'$, and hence EF{1} follows by induction. If $s = p$, then $p$ is a source in $G_{A'}^j$. Let us now consider two subcases for the allocation $A'$: (i) $r \notin N_j$, and (ii) $r \in N_j$. In Subcase (i), recalling that the number of items of type $j$ in a bundle never exceeds $\theta_j$ in Phase 1, we have $a_{rj}$ $=a'_{rj}$ $= \theta_j$. Then $v_r(A_r)$ $= v_r(\hat{A}_r) + v_{rj}(\theta_j)$ $< v_r(A_p)$ $= v_r(\hat{A}_p) + v_{rj}(A_{pj})$ $< v_r(\hat{A}_p) + v_{rj}(\theta_j)$. Hence, $r$ envies $s$ even in $\hat{A}$, i.e., prior to type $j$ assignments. Since the lemma holds for $\hat{A}$, $\exists x \in \hat{A}_p$ such that $v_r(\hat{A}_p \setminus \{x\})$ $\leq v_r(\hat{A}_r)$. Then, $v_r(A_p \setminus \{x\})$ $\le v_r(\hat{A}_p \setminus \{x\}) + v_{rj}(\theta_j)$ $\le v_r(\hat{A}_r) + v_{rj}(\theta_j)$ $= v_r(A_r)$, and hence removing item $x$ eliminates the envy in allocation $A$ as well. 

In Subcase (ii), the envy from agent $r$ to $s$ can be eliminated by removing the newly assigned item. Hence the lemma holds inductively.
\end{proof}

\begin{lemma}
    Let $A$ be the (partial) allocation at any point in Phase 2 (i.e., in the second for all loop). Then if agent $r$ envies agent $s$, then the envy can be resolved by removal of a single item from either of the bundles. 
\label{lem:identThresh2}
\end{lemma}
\begin{proof}
Fix a type $j$. Note that at the end of Phase 1, either $\hat{m}_j = 0$, or for each agent $i$, $a_{ij} = \theta_j$. 

    The proof is now by induction. At the beginning of Phase 2, we have an EF{1} allocation that is obtained post Phase 1. Hence, the statement holds true. It can be easily verified that the elimination of top-trading envy cycles preserves EF1 and also guarantees a sink in the envy graph. Suppose $A$ is the allocation obtained after allotment of an item of type $j$ and suppose $A'$ is the allocation just before the allotment of the item. By induction hypothesis, the statement holds for $A'$. Suppose agent $r$ envies agent $s$ in $A$. If $r$, $s \neq k$, then the statement trivially holds as in $A'$. We have to examine the rest of the cases.
    \begin{enumerate}
        \item $r = k$. Since $k$ was a sink in $G_{A'}$, she did not envy $s$ in $A'$. Hence this envy can be resolved by removing the most recently allocated item.
        \item $s = k$. Note that both $a'_{rj}$ and $a'_{sj}$ are at least $\theta_j$. Hence if $r$ envies $s$ in allocation $A$, this must be true of allocation $A'$ as well, since the value of $s$'s bundle has decreased for $r$. We now consider two subcases here, depending on whether $a'_{sj} > \theta_j$ or $a'_{sj} = \theta_j$. In Subcase (i), if $a'_{sj} > \theta_j$, then removal of an item of type $j$ from $A'_s$ would only increase the bundle's value for any agent. Hence, by the induction hypothesis, since $v_r(A_s') > v_r(A_r')$, the envy is resolved by either removing an item of type $j' \neq j$ from $s$, or by removing an item from agent $r$. In either case, this continues to hold true for allocation $A$ as well.
        
        In Subcase (ii), $a'_{sj} = \theta_j$ and $a'_{rj} \ge \theta_j$. 
        Consider the last step where $r$ did not envy $s$, and let $\hat{A}$ be this allocation. Thus, $v_r(\hat{A}_r) \ge v_r(\hat{A}_s)$. If this last step was in Phase 2, then immediately following this step, agent $r$ receives an item $x$ and envies $s$. Removing this item $x$ from agent $r$'s bundle clearly removes the envy. Further, following this step, the value of agent $s$'s bundle can only decrease for $r$ (since all items are added past the threshold). Since $r$ continues to envy $s$, agent $r$ receives no more items. Hence removing item $x$ from $r$'s bundle removes any envy.
        
        Now assume the last step where $r$ did not envy $s$ was in Phase 1, either before or after type $j$ items were assigned. Then immediately following this last step, agent $s$ was assigned an item $x$ of type $j' \neq j$, and the removal of this item $x$ removes the newly created envy. Further, since $r$ envies $s$ following the allocation of item $x$, agent $r$'s value for her own bundle has not decreased (since $r$ has never been a sink), and agent $r$'s value for $s$'s bundle has not increased (since $s$ has never been a source). Hence, removing item $x$ from the current bundle $A$ also removes envy, i.e., $v_r(A_r) \ge v_r(A_s \setminus \{x\})$.
        
        Lastly, suppose the last step where $r$ did not envy $s$ was in Phase 1, during the assignment of items of type $j$. Then $r$ envies agent $s$ after type $j$ items are assigned in Phase 1 as well. Since both $r$ and $s$ get exactly $\theta_j$ items of type $j$ in Phase 1, this envy exists before any items of type $j$ are assigned in Phase 1. Hence, as previously, there exists $x \in A_s$ of type $j' \neq j$, the removal of which removes envy.
    \end{enumerate}
\end{proof}

\end{document}